\pdfoutput=1
\documentclass[11pt]{article} 
 
\usepackage[hyphens]{url}  
\usepackage{geometry}
\usepackage{graphicx} 
\usepackage{natbib}  
\usepackage{caption} 

\usepackage{xcolor}
\definecolor{bleudefrance}{rgb}{0.19, 0.55, 0.91}
\usepackage{algorithm}
\usepackage{algorithmic}

\usepackage{newfloat}
\usepackage{listings}
\DeclareCaptionStyle{ruled}{labelfont=normalfont,labelsep=colon,strut=off} 
\floatstyle{ruled}
\newfloat{listing}{tb}{lst}{}
\floatname{listing}{Listing}
\usepackage[utf8]{inputenc} 
\usepackage[T1]{fontenc}    

\usepackage{hyperref}  
\usepackage{amsthm,mathtools}
\usepackage{thmtools}

\hypersetup{
	colorlinks   = true,
	citecolor    = bleudefrance,
	linkcolor	  = red,
	urlcolor     = bleudefrance
}
\usepackage[capitalize, nameinlink]{cleveref}
\crefformat{equation}{(#2#1#3)} 

\usepackage{url}            
\usepackage{booktabs}       
\usepackage{amsfonts}       
\usepackage{nicefrac}       
\usepackage{microtype}      

\usepackage{xcolor}         
\usepackage{comment}        
\usepackage{caption}

\theoremstyle{plain}
\newtheorem{theorem}{Theorem}

\newtheorem{lemma}[theorem]{Lemma}
\newtheorem{proposition}[theorem]{Proposition}

\theoremstyle{definition}
\newtheorem{definition}[theorem]{Definition}

\newtheorem{example}[theorem]{Example}

\usepackage{amsmath,amsfonts,bm}

\def\eqref#1{equation~\ref{#1}}

\def\1{\bm{1}}

\def\eps{{\epsilon}}

\def\rvz{{\vec{Z}}}

\def\vzero{{\mathbf{0}}}

\def\vb{{\vec{b}}}
\def\vc{{\vec{c}}}

\def\ve{{\vec{e}}}

\def\vu{{\vec{u}}}
\def\vv{{\vec{v}}}
\def\vw{{\vec{w}}}
\def\vx{{\vec{x}}}
\def\vy{{\vec{y}}}
\def\vz{{\vec{z}}}

\def\mA{{\bm{A}}}
\def\mB{{\bm{B}}}
\def\mC{{\bm{C}}}
\def\mD{{\bm{D}}}

\def\mQ{{\bm{Q}}}

\def\mU{{\bm{U}}}
\def\mV{{\bm{V}}}

\DeclareMathAlphabet{\mathsfit}{\encodingdefault}{\sfdefault}{m}{sl}
\SetMathAlphabet{\mathsfit}{bold}{\encodingdefault}{\sfdefault}{bx}{n}

\newcommand{\E}{\mathbb{E}}

\newcommand{\R}{\mathbb{R}}
\newcommand{\N}{\mathbb{N}}
\newcommand{\Z}{\mathbb{Z}}

\newcommand{\A}{\mathcal{A}}

\DeclareMathOperator{\vol}{Vol}

\newcommand\fang[1]{\textcolor{olive}{}}
\newcommand\prat[1]{\textcolor{cyan}{}}
\newcommand\robin[1]{\textcolor{blue}{}}
\newcommand\jie[1]{\textcolor{red}{}}

\title{Integer Subspace Differential Privacy}

\date{}

\author{
    Prathamesh Dharangutte\thanks{Rutgers University, \texttt{ptd39@rutgers.edu}} \and
    Jie Gao\thanks{Rutgers University, \texttt{jg1555@rutgers.edu}} \and
    Ruobin Gong\thanks{Rutgers University, \texttt{ruobin.gong@rutgers.edu}} \and
    Fang-Yi Yu\thanks{George Mason University, \texttt{fangyiyu@gmu.edu}}
}

\begin{document}

\maketitle

\begin{abstract}
We propose new differential privacy solutions for when external \emph{invariants} and \emph{integer} constraints are simultaneously enforced on the data product. These requirements arise in real world applications of private data curation, including the  public release of the 2020 U.S. Decennial Census. They pose a great challenge to the production of provably private data products with adequate statistical usability. We propose \emph{integer subspace differential privacy} to rigorously articulate the privacy guarantee when data products maintain both the invariants and integer characteristics, and demonstrate the composition and post-processing properties of our proposal. To address the challenge of sampling from a potentially highly restricted discrete space, we devise a pair of unbiased additive mechanisms, the generalized Laplace and the generalized Gaussian mechanisms, by solving the Diophantine equations as defined by the constraints. The proposed mechanisms have good accuracy, with errors exhibiting sub-exponential and sub-Gaussian tail probabilities respectively. To implement our proposal, we design an MCMC algorithm and supply empirical convergence assessment using estimated upper bounds on the total variation distance via $L$-lag coupling. We demonstrate the efficacy of our proposal with applications to a synthetic problem with intersecting invariants, a sensitive contingency table with known margins, and the 2010 Census county-level demonstration data with mandated fixed state population totals.
\end{abstract}

\section{Introduction}\label{sec:intro}

{\bf Motivation.}
Differential privacy (DP) is a formal mathematical framework that quantifies the extent to which an adversary can learn about an individual from sanitized data products. 
However, data curators may be mandated, by law or by policy, to release sanitized data products that obey certain externally determined constraints, in a manner that existing differential privacy solutions are not designed to address. Two challenging types of constraints are \emph{invariants} \citep{ashmead2019effective}, which are exact statistics calculated from the confidential data (such as sum, margins of a contingency table, etc), and \emph{integer} characteristics, pertaining to data of count nature. The challenge is exacerbated when the sanitized data product is required to respect both. 


Data privatization that simultaneously satisfies mandated invariants and preserves integral characteristics is of paramount importance to a myriad of data products by official statistical agencies.
One of the prominent example is 
the Disclosure Avoidance System (DAS) of the 2020 Decennial Census of the United States. The Census Bureau is  mandated to observe a set of invariants for its decennial census data products, such as its constitutional obligation to enumerate state populations for the apportionment house seats. 
As the decennial census data products are count data, the bureau must ensure that the sanitized release are also integer-valued. In addition, the bureau strives to maintain the consistency of the data products with common knowledge to various degrees, such as the the non-negativity of count data and possible logical relationships between different tabulated quantities.
The challenge thus remains: how to 
design sanitized data products that respect the pre-specified constraints, while maintaining rigorous privacy and good statistical properties? This is the central problem we address in this paper.

{\bf Our contribution.} 
In this paper,  we develop the \emph{integer subspace differential privacy} scheme, through which we formulate and implement mathematically rigorous privacy mechanisms that meet the following \emph{utility objectives}: 1) respecting mandated invariants, 2) maintaining integral characteristics, and 3) achieving good statistical properties including \emph{unbiasedness}, \emph{accuracy}, as well as  \emph{probabilistic transparency} of the mechanism specification.  
Invariants not only restrict the universe of databases of interest, but also partially destroys the relevance of neighbors, since databases differing by precisely one entry may or may not meet the same invariants. To this end, we adapt the classical differential privacy definition to incorporate invariant constraints with two modifications (i)
we limit the comparison of databases only to those that satisfy the same invariants, in a manner similar to pufferfish privacy~\citep{kifer2014pufferfish}; and (ii) 
for databases that meet the same invariants, we use their metric distance 
as a multiplicative factor to the privacy loss parameter under the more general smooth differential privacy framework~\citep{chatzikokolakis2013broadening,damien2019sok}. 
We show that integer subspace differential privacy preserves  \emph{composition} and \emph{post-processing} properties. 

To design a differential privacy mechanism that simultaneously respects invariants and integer requirement  poses a number of new challenges.  For additive mechanisms, 
the perturbation noise vector needs to be limited to the ``null space'' of the invariants so as not to violate the constraints. When the invariants are linear, this requires solving a linear system with integer coefficients, i.e. Diophantine equations, limiting the privacy noise to a (transformed) discrete lattice space. 
To this end, we propose generalized discrete Laplace and Gaussian mechanisms for the discrete lattice space. Both mechanisms are transparently specified additive mechanisms that enjoy demonstrated unbiasedness and accuracy via bounds on tail probabilities. The tail bounds
call for extra technicality as the lattice spaces we deal with are generally not spherically symmetric. To implement the proposed mechanisms, we design a Gibbs-within-Metropolis sampler, with a transition kernel that is always limited to the desired subspace to achieve sampling efficiency. To provide empirical guidance on the Markov chain's convergence to target, we supply an assessment scheme using the $L$-lag coupling method proposed by~\citep{biswas2019estimating}. We demonstrate the efficacy of our proposal with applications to a synthetic problem with intersecting invariants, a contingency table with known margins concerning delinquent children~\citep{fcsm2005sdl}, and the 2010 Census county-level demonstration data with mandated fixed state population totals.


\textbf{Related work.}
There has been an extensive line of work on differential privacy and its alternative definitions; see a recent survey and the references therein~\citep{damien2019sok}. On the other hand, invariants pose a relatively recent challenge to formal privacy, and we focus our review below on prior work concerning invariants. It has been recognized that invariants which are non-trivial functions of the confidential data compromise the differential privacy guarantee in its classic sense, and that the invariants must be incorporated into the privacy guarantee \citep{gong2020congenial,seeman2022partially}. Towards designing formally private mechanisms that meet prescribed invariant constraints, the only prior work in this direction is named \emph{subspace differential privacy}~\citep{gao2022subspace} where the inputs satisfy a set of linear constraints and the sanitized outcome shall meet the same invariants. 
However, subspace differential privacy only considers data products that are real-valued. To impose the additional integral requirement on the data outputs is not trivial: as the TopDown algorithm \citep{abowd2022topdown} already demonstrates, simple rounding or projection will either violate the invariants or introduce biases.
\citet{he2014blowfish} consider known constraints of the dataset (e.g. population counts for each state) which restricts the possible set of datasets, but the output of their mechanism need not satisfy the known constraints. We focus on mechanisms whose output also satisfy the counting constraints. Another line of work \citep{cormode2017constrained} consider additional structural properties for the mechanism's output (e.g. monotone, symmetric, fair). However, those properties are independent of the private data and can be seen as internal distributional consistency.
Lastly we note a related, but different line of previous work concerning \emph{internal consistency} of the data outputs~\citep{barak2007privacy,hay2009boosting}, such as maintaining the sum of counts over some partitioning of a set to be identical as the total sum. This type of consistency requirement is independent of the value of the private data, and can be satisfied by certain types of DP mechanisms, such as local DP schemes when each data item is independently perturbed and normal algebraic operations are applied on the perturbed output. In contrast, invariant constraints considered in this work are in general non-trivial functions of the confidential data. 


To impose invariants and preserve integrality for Census data, the TopDown algorithm~\citep{abowd2022topdown} employs a discrete Gaussian mechanism \citep{canonne2020discrete} to inject integer-valued noise during the so-called \emph{measurement} phase, to form noise-infused counts to be tabulated into multi-way contingency tables at various levels of geographic resolution.
This is followed by the so-called \emph{estimation} phase, a constrained optimization to impose invariants (such as state population totals, etc) and controlled rounding to integer solutions \cite[see Table 1 of][]{abowd2022topdown}. In general, this approach first imposes unconstrained noise and performs post-processing to meet the additional constraints. One limitation of this solution is that it is not easy to characterize the  probabilistic distribution of the resulting privacy noise after post-processing, as this distribution crucially depends on the particular values of the input data. Thus, the post-processing approach may destroy the probabilistic transparency of the privatized data product \citep{gong2022transparent}.
Our solution maintains the  transparency of the privacy mechanism as well as the statistical independence between the confidential data and the privacy errors, which are of paramount importance to downstream applications derived from these noisy data products. Furthermore, since invariants and integer value constraints are our first priority besides privacy, we supply a rigorously updated definition of differential privacy under this setting by limiting the discussion only among databases that satisfy these constraints.




\section{Integer Subspace Differential Privacy}

In this work, we model private data as a histogram $\vx=\left(x_{1},\dots,x_{d}\right)^{\top} \in  \mathbb{N}^d$, where $x_i\in \mathbb{N}$ is the number of individuals with feature $i\in [d] := \{1, \dots, d\}$.  
A trusted curator holds the database $\vx$, and provides an interface to the database through a randomized mechanism $M:\mathbb{N}^d\to \mathbb{Z}^d$ where we require the output to always take integer values. As motivated in introduction, our goal is to design good mechanisms whose output is close to the original histogram that satisfies not only certain privacy notions, but also invariant constraints and integral requirements.



The notion of differential privacy ensures that no individual's data has much effect on the probabilistic characteristic of the output from the mechanism $M$~\citep{dwork2006calibrating,https://doi.org/10.48550/arxiv.1412.4451}.  
Formally, we say that a randomized mechanism $M: \N^d\to\Z^d$ is \emph{$\left(\epsilon,\delta\right)$-differentially private} if for some $\epsilon, \delta\ge0$, all neighboring databases $\vx$ and $\vx'$ with $\|\vx-\vx'\| = 1$, and all event $S\subseteq \Z^d$, 
\begin{equation}\label{eq:dp-pure-definition}
\Pr\left[ M\left(\vx\right)\in S\right]\le e^\epsilon\cdot \Pr\left[M\left(\vx'\right)\in S\right]+\delta,    
\end{equation}
where $\|\vx-\vx'\|$ is a distance norm of two databases $\left(\vx, \vx'\right)$. Depending on the application, the norm may be chosen either as $\ell_1$ norm or $\ell_2$ norm.

The notion of differential privacy in~\cref{eq:dp-pure-definition} satisfies a more general notion of \emph{smooth differential privacy}~\citep{chatzikokolakis2013broadening}.  $M$ is $(\epsilon, \delta)$-\emph{smooth differentially private} if for all datasets $\vx, \vx'$ and all event $S$ ,
\begin{equation}\label{eq:dp-approx-definition}
\Pr[M(\vx) \in S]\le e^{\epsilon\|\vx-\vx'\|}\Pr[M(\vx') \in S] +\frac{e^{\epsilon\|\vx-\vx'\|}-1}{e^{\epsilon}-1}\delta.  
\end{equation}
Previous works mostly consider the pure smooth differential privacy with $\delta = 0$, and we generalize it to the approximated case.\footnote{When $\delta = 0$, \cref{eq:dp-approx-definition} reduces to $\Pr[M(\vx) \in E]\le e^{\epsilon\|\vx-\vx'\|}\cdot \Pr[M(\vx') \in E]$.  However, for the approximated smooth differential privacy, the error term grows as the distance $\|\vx-\vx'\|$ increases.  For instance, as the distance increases, we can see the error term of the Gaussian mechanism also grows.  Finally, when $\|\cdot\|$ is a norm, a mechanism on the histogram is $(\epsilon,\delta)$-differentially private if and only if the mechanism is $(\epsilon, \delta)$-smooth differentially private.}
Compared to~\cref{eq:dp-pure-definition}, the definition in~\cref{eq:dp-approx-definition} replaces the neighboring relationship of $\left(\vx, \vx'\right)$ by a norm function.  

From the data curator's perspective, in addition to privacy concerns, there often exists external constraints that the privatized output $M$ must meet. 
These constraints can often be represented as counting functions defined below.

\begin{definition}[name = counting invariant constraints, label = def:count_inv]
Given a collection of subsets $\mathcal{A} = \{A_1,\dots,A_k\}$ on $[d]$ with $k\le d$, we say a function $f:\N^d\to \Z^d$ is $\mathcal{A}$-invariant if
$$\sum_{i\in A} x_i = \sum_{i\in A} f(\vx)_i\text{, for all }\vx\in \N^n\text{ and }A\in \mathcal{A}.$$
\end{definition}
Alternatively, given a collection of counting invariant constraints $\mathcal{A}$, we define an equivalence relationship $\equiv_\mathcal{A}$ so that $\vx\equiv_\mathcal{A}\vx'$ if $\sum_{i\in A}x_i = \sum_{i\in A}x_i'$ for all $A\in \mathcal{A}$.  Then an $\A$-invariant $f$ satisfies $\vx\equiv_\A f(\vx)$ for all $\vx$.


Note that when a privatized output is subject to counting invariant constraints as above, we may not use the standard neighboring relation because two neighboring datasets differing by one element may not satisfy the same counting constraint.  In particular, the feasible outputs of an $\A$-invariant mechanism lie within an integer subspace.  Therefore, we employ smooth differential privacy, and only control the privacy loss on datasets that satisfy the same counting constraints, i.e., only ``secret pairs'' of databases within the same equivalent classes. This formulation was also used in previous work such as pufferfish privacy \citep{kifer2014pufferfish}.

\begin{definition}
\label{def:partitional_dp}
We say a mechanism $M$ is 
$(\epsilon, \delta)$-differentially private on an equivalence relation $\equiv$ if for all equivalent pair $\vx\equiv \vx'$ and event $E$ on the output space
$$\Pr[M(\vx) \in E]\le e^{\epsilon\|\vx-\vx'\|}\Pr[M(\vx') \in E] +\frac{e^{\epsilon\|\vx-\vx'\|}-1}{e^{\epsilon}-1}\delta.$$ 

Moreover, given a collection of counting constraint $\A$, we say $M$ is \emph{$\A$-induced integer subspace differentially private} with \emph{privacy loss budget} $(\epsilon, \delta)$, if $M$ is $\A$-invariant and $(\epsilon, \delta)$-differentially private on an equivalence relation $\equiv_\A$.

\end{definition}


Many invariant constraints can be formulated using counting constraints. Below are some examples.
\begin{example}
Suppose $\A = \{[d]\}$ is a singleton.  All datasets with the same number of agents are equivalent under $\equiv_{\{[d]\}}$, so the differential privacy on $\equiv_{\{[d]\}}$ reduces to the original (bounded) differential privacy. 
\end{example}

\begin{example}
Privatized census data products must ensure that the total population of each state is reported exactly as enumerated.
Given $k$ states, and $A_l\subset[d]$ be the collection of features that belong to state $l$. we can define counting invariants $\A = \{A_1,\dots, A_k\}$ where $A_1,\dots,A_k$ forms a partition on the universe $[d]$. 
\end{example}

\begin{example}
We can encode the invariant condition on a $\sqrt{d}\times \sqrt{d}$ two-dimensional contingent table where the sum of each row and each column sums are fixed as a collection of counting constraints.   Formally, given $\sqrt{d}\in \mathbb{N}$, the set $\A$ contains the following subsets on $[d]$, for all $i$ and $j$
$R_{i} = \{i\sqrt{d}+\ell:0\le \ell<\sqrt{d}\}\text{ and }C_{j} = \{j+\ell\sqrt{d}:0\le \ell<\sqrt{d}\}.$
That is $\A = \{R_i, C_j:0\le i,j< \sqrt{d}\}$. 
\end{example}




Composition and post-processing properties for \cref{def:partitional_dp} also hold due to similar arguments as the differential privacy\fang{why do we use approximate dp here?}, but with the additional requirement of equivalence over databases.
Let $\equiv_{(\A_1,\A_2)}$ be the equivalence relation such that $x \equiv_{(\A_1,\A_2)} x'$ implies $x \equiv_{\A_1} x'$ and $x \equiv_{\A_2} x'$. 

\begin{proposition}[Composition]
\label{prop:composition}
Let $M_1$ be $(\epsilon_1, \delta_1)$- differentially private on equivalence relation $\equiv_{\A_1}$ and $M_2$ be $(\epsilon_2, \delta_2)$- differentially private on equivalence relation $\equiv_{\A_2}$. For any pair of databases $x, x'$ having $x \equiv_{(\A_1,\A_2)} x'$, the composed mechanism $M_{1,2}(x) = (M_1(x),M_2(x))$ is $(\epsilon_1+\epsilon_2, \delta_1 + \delta_2)$-differentially private on $\equiv_{\A_1, \A_2}$.
\end{proposition}



\begin{proposition}[Post-processing]
\label{prop:post-processing}
If $M$ is $(\epsilon,\delta)$-differentially private over $\equiv$, then for a arbitrary randomized mapping $F: \mathcal{Y} \to \mathcal{Z}$,  $F \circ M$ is $(\epsilon,\delta)$-differentially private over $\equiv$.
\end{proposition}


Details of ~\cref{prop:composition}~and~\ref{prop:post-processing} can be found in ~\cref{appendix:composition-post-processing}. Note that when privacy mechanisms that obey different invariant constraints are imposed, it might be possible to infer aspects of the confidential data following logical consequences from both. 
It is not clear how to define post-processing for invariant constraints in those situations.

\section{Generalized Laplace and Gaussian Mechanisms}
Now we study $\A$-induced integer differentially private mechanisms that satisfy counting invariant constraints.  The main challenge is that counting invariant constraints significantly restrict the feasible output space, especially when the output is required to take only integer values.   To resolve this issue, we first show that counting invariant constraints on integer-valued output can be written as a lattice space.  Then we propose revised Laplace and Gaussian mechanisms on a lattice space. We discuss how to implement these mechanisms via sampling and MCMC in the Implementation section.


\subsection{Counting invariants and lattice spaces}\label{sec:feasible}
In this section, we show the output space of integer datasets satisfying counting invariant constraints are lattice spaces. 
A \emph{lattice} $\Lambda$ is a discrete additive subgroup of $\mathbb{R}^d$.  Given a matrix $\mB\in \R^{d\times m}$ consisting of $m$ basis $\vb_1,\dots, \vb_m$ with $1\le m\le d$, a lattice generated by the basis is 
$\Lambda(\mB) = \left\{\mB\vv = \sum_{i = 1}^m v_i\vb_i: \vv\in \mathbb{Z}^m\right\}.$

\begin{proposition}
\label{prop:feasible}
Given a collection of counting invariants $\A=\{A_1,\dots,A_k\}$, there exists a lattice $\Lambda_\A = \Lambda(\mC_\A)$with basis $\mC_\A\in \Z^{d\times(d-k)}$ so that a function $f:\N^d\to \Z^d$ is $\A$-invariant if and only if for all $\vx\in \N^d$
$$f(\vx)-\vx\in \Lambda_\A.$$
\end{proposition}
Note that as $\mC_\A$ is an integer-valued matrix, $\Lambda_\A$ is a subset of the integer grid $\Z^{d}$.  As we impose more counting invariants, the lattice $\Lambda_\A$ becomes sparser, making it harder to find a feasible solution for an $\A$-invariant mechanism. 

\begin{proof}\fang{If we need more space, I feel we may remove this proof which is mostly linear algebra}
Given a collection of counting invariant constraints $\A=\{A_1,\dots,A_k\}$ and a dataset $\vx$, finding a feasible output $\vy$ with $\vy\equiv_\A \vx$ is equivalent to solving the following linear equations in \cref{eq:lin_sys}, where $\mA\in\{0,1\}^{k\times d}$ is an \emph{incidence matrix}, for all $i\in[d]$ and $l\le k$ $A_{l,i} = 1$ if $i\in A_l$ and zero otherwise.  Then $\vy\equiv_\A \vx$ if and only if 
\begin{equation}\label{eq:lin_sys}
    \mA\vy = \mA\vx.
\end{equation}
Since we require $\vy\in \Z^d$, the problem of solving $\vy$ is known as solving the linear Diophantine equation~\citep{gilbert1990linear,schrijver1998theory,greenberg1971integer}, and can be done by computing the \emph{Smith normal form} of $\mA$.  Specifically, given an integer-valued matrix $\mA\in \Z^{k\times d}$ of rank $k$, there exists $\mU\in \Z^{k\times k}, \mV\in \Z^{d\times d}$ and $\mD\in \Z^{k\times d}$ with 
\begin{equation}\label{eq:smith-normal-form}
\mU\mA\mV = \mD \quad 
\end{equation}
so that $\mU$ and $\mV$ are unimodular matrices that are invertible over the integers, and $\mD$ is a diagonal matrix (i.e. the Smith normal form of $\mA$) with $D_{l,l}\neq 0$ for all $l\in[k]$. \fang{no need to specify off-diagonal entries} 

With the above decomposition, we can characterize the integer solutions of \cref{eq:lin_sys}. $\vy$ is a solution of \cref{eq:lin_sys}, if and only if  $\mA (\vy-\vx) = \vzero$ and it is equivalent to $\mD\mV^{-1} (\vy-\vx) = \vzero$, since $\mU$ is unimodular.  Because $\mD$ is diagonal and has rank $k$, $\vy$ is a solution if and only if there exists $\vw\in \Z^d$ so that $\vy = \vx+\mV\vw$ and $w_l = 0$ for all $l\le k$.
Additionally, if we define $\mC_\A\in\Z^{d\times (d-k)}$ that consists the bottom $d-k$ rows of $\mV$, then $\vy$ is a solution if and only if $\vy- \vx \in\{\mC_\A\vv:\vv\in \Z^{d-k}\} =  \Lambda(\mC_\A)$.
We call $\A$ \emph{full rank} if the rank of $\mA\in \{0,1\}^{k\times d}$ is $k$, the associated $\mC_\A$ as the \emph{basis of $\A$}, and $\Lambda_\A:=\Lambda(\mC_\A)$.  The integer solutions of \cref{eq:lin_sys} is the lattice $\Lambda_\A$ shifted by $\vx$.  Therefore, if $M:\N^d\to \Z^d$ is $\A$-invariants the output given input $\vx$ is contained in $\vx+\Lambda_\A := \{\vx+\vz:\vz\in \Lambda_\A\}$,
$M(\vx)\in \vx+\Lambda_\A.$
\end{proof}




\subsection{Generalized Laplace mechanism}
Now we can define a generalized Laplace mechanism whose output space satisfies \cref{prop:feasible}.  We adopt the classical exponential mechanism to this restricted output space and show the resulting mechanisms are as required. Because the output space \cref{prop:feasible} is unbounded, we show the utility guarantee by controlling the dimension of the lattice spaces which may be of interest by itself. First, we formally define the generalized Laplace mechanism.

\begin{definition}[label = def:lap]
Given a collection of counting invariant condition $\A$ of full rank and $\epsilon>0$, the \emph{generalized Laplace mechanism} is $M_{Lap,\A, \epsilon}(\vx) = \vx+\vz$ with $\vz$ sampled from \begin{equation}\label{eq:gen-laplace}
  q_\epsilon(\vz)\propto \exp\left(-\epsilon\|\vz\|\right)\mathbf{1}[\vz\in \Lambda_\A]  
\end{equation}
where $\Lambda_\A$ is defined in \cref{prop:feasible}.
\end{definition}
Note that if the counting constraint is \emph{vacuous}, i.e. $\A = \emptyset$ with $d = 1$, and the output can be real number, the generalized Laplace mechanism reduces to the original Laplace mechanism on histograms.  Additionally, if we require integer-valued output with $\A = \emptyset$, the generalized Laplace mechanism becomes the double geometric mechanism. 

\begin{theorem}[label = thm:lap]
Given full rank $\A$ and $\epsilon>0$, mechanism $M_{Lap,\A, \epsilon}:\N^d\to \Z^d$ in \cref{def:lap} is $\A$-induced integer subspace  with $(\epsilon,0)$.
\end{theorem}
The proof is similar to the privacy guarantee of Laplace mechanisms, and is presented in ~\cref{appendix:lattice}. The implementation of the generalized Laplace mechanisms in \cref{def:lap} is non-trivial and is elaborated in Implementation section. 

After the privacy guarantee, we turn to discuss the utility of the generalized Laplace mechanism.  First, \cref{eq:lap_unbias} establishes the unbiasedness of the generalized Laplace mechanism. The tail bound needs additional work, because the output spaces of the conventional exponential mechanism are often finite, which is not the case here. 
We resolve this issue using the bounded-dimension property of lattices spaces (\cref{lem:doubling}), and prove the error bound is sub-exponential in \cref{eq:lap_tail}.  
There is a long line of work~\citep{landau1915analytischen,tsuji1953lattice,bentkus1999lattice} that approximate the number of lattice points in sphere by the volume.  Using these ideas, we prove the following lemma (with proof in ~\cref{appendix:lattice}) for ~\cref{thm:lap_acc}.
\begin{lemma}[label = lem:doubling]
Let $\mB\in \R^{d\times m}$ be a rank $m\le d$ matrix.  For all $r>0$ large enough, 
$|\{\vz\in \Lambda(\mB):\|\vz\|_2\le r\}|\le \frac{2V_m}{\sqrt{\det(\mB^\top \mB)}}r^m$
where $V_{m}$ is the volume of the $m$ dimensional unit sphere. 
\end{lemma}
\begin{theorem}[name = unbiasedness and accuracy,label = thm:lap_acc]
Given $\A$ of $k$ counting constraints and $\epsilon>0$, $M_{Lap, \A,\epsilon}$ in \cref{def:lap} is unbiased 
\begin{equation}\label{eq:lap_unbias}
    \E[M_{Lap,\A,\epsilon}(\vx)] = \vx,
\end{equation}
and, if $\|\cdot\|$ is $\ell_2$-norm, there exists a constant $K>0$ so that for all $\vx\in \N^d$, and large enough $t>0$, 
\begin{equation}\label{eq:lap_tail}
    \Pr[\|M_{Lap,\A,\epsilon}(\vx)-\vx\|_2\ge t]\le Kt^{d-k}\exp(-\epsilon t).
\end{equation}
Moreover, $K$ can be $\frac{4\cdot 2^{d-k}V_{d-k}}{\sqrt{\det(\mC_\A^\top\mC_\A)}}$ where $V_{d-k}$ is the volume of the $(d-k)$ dimensional unit sphere.
\end{theorem}
Note that as the dimension $d-k$ increases $2^{d-k}V_{d-k}\to 0$, so $K$ is determined by the counting constraints $1/\sqrt{\det(\mC_\A^\top\mC_\A)}$.


\subsection{Generalized Gaussian mechanism}
In this section, we propose generalized Gaussian mechanisms that relate to Lattice-based cryptography discussed in \cref{sec:implementation}.  
We first define Gaussian random variables on lattices. Given a lattice $\Lambda\subset\R^d$, and $\vc\in \R^d$, the spherical \emph{Gaussian random variable} $\rvz$ on $\Lambda$ with variance $\sigma$ and center $\vc$ is 
\begin{equation}\label{eq:gen-gaussian}
\Pr[\rvz = \vz]\propto \exp\left(-\frac{1}{2\sigma^2} \|\vz-\vc\|_2^2\right), \forall \vz\in \Lambda.
\end{equation}
Note that the generalized Gaussian mechanism utilizes an $\ell_2$ norm to measure the noise scale.


\begin{definition}[label = def:gaussian]
Given a collection of counting invariant condition $\A$ of full rank, $\epsilon>0$, and $\delta>0$, the \emph{generalized Gaussian mechanism} is defined as $M_{G,\A, \epsilon, \delta}(\vx) = \vx+\vz$ where $\vz$ is the Gaussian random variable on $\Lambda_\A$ with center at $\vzero$ and variance $\sigma_{\epsilon, \delta}^2 = \frac{2c_\A \ln 1/\delta}{\epsilon^2}$ for some constant $c_\A = \mathcal{O}(\max\{(d-k)\ln (d-k), \ln K\} )$ that only depends on the dimension and the collection of counting constraints where $K$ is defined in \cref{thm:lap_acc}. 
\end{definition}
While the variance scale in the original Gaussian mechanisms is independent of the dimension, generalized Gaussian mechanism's $\sigma_{\epsilon, \delta}$ in \cref{def:gaussian} depends on the dimension.  This may be due to a lack of symmetry of the lattice space.  Recall that the proof of the original Gaussian mechanism reduces the high dimensional case to the one dimensional setting, thanks to the Gaussian being spherically symmetric.  However, our lattice space $\Lambda_\A$ is generally not spherically symmetric, and simple reduction does not work.

\begin{theorem}[label = thm:gaussian]
Given full rank $\A$, $\epsilon$ and $\delta$ with $0<\delta<\epsilon<1/e$, mechanism $M_{G,\A, \epsilon, \delta}:\N^d\to \Z^d$ in \cref{def:gaussian} is $\A$-induced integer subspace differentially private with $(\epsilon,\delta)$.
\end{theorem}

After the privacy guarantee, we turn to show the utility of the generalized Gaussian mechanism.  First, since the distribution function of the discrete Gaussian is an even function, the additive errors of the discrete Gaussian mechanism is unbiased~\cref{eq:gaussian_unbias}.  Furthermore,  similar to Gaussians on the Euclidean space, Gaussians on lattices are sub-Gaussian.\footnote{This is nontrivial, because the support of Gaussian in \cref{eq:gen-gaussian} is not uniform.  For instance, suppose $\Lambda$ is not a lattice space, but has $2^{r^2}$ points in each integer radius $r$ shell, the a Gaussian on such a space is not sub-Gaussian.}


\begin{theorem}[name = unbiasedness and accuracy]
Given the condition in \cref{thm:gaussian}, $M_{G,\A, \epsilon, \delta}$ in \cref{def:gaussian} is unbiased,
\begin{equation}\label{eq:gaussian_unbias}
    \E[M_{G,\A, \epsilon, \delta}(\vx)] = \vx \text{, for all }\vx\in \N^d.
\end{equation}
Additionally, there exists a constant $K>0$ defined in \cref{thm:lap_acc} so that for all $\vx\in \N^d$, and large enough $t>0$, 
\begin{equation}\label{eq:gaussian_tail}
    \Pr[\|M_{G,\A,\epsilon, \delta}(\vx)-\vx\|_2\ge t]\le Kt^{d-k}e^{-\frac{t^2}{2\sigma_{\epsilon, \delta}^2}}.
\end{equation}
\end{theorem}

Last, we remark that previous work has used discrete Gaussian mechanisms~\citep{canonne2020discrete}for integer valued noises. But these mechanisms do not satisfy (non-trivial) invariants constraints. 
Note that when there is no invariant constraint, the resulting lattice space becomes an integer grid and our generalized Gaussian mechanism reduces to the discrete Gaussian mechanism.

\section{Implementation}\label{sec:implementation}






We discuss implementation issues of discrete Gaussian and generalized Laplace in lattice spaces. In general the design of exact samplers for complex privacy mechanisms (such as the Exponential mechanism, to which our generalized Laplace mechanism is a special case) is a widely acknowledged challenge in the literature. Section 3 of Abowd et al.~\citep{abowd2022topdown} specifically discussed this challenge in the context of invariants and integer constraints and concluded that direct sampling from the Exponential mechanism is ``infeasible'' for their TopDown algorithm. We address this problem in two ways. For certain discrete Gaussian, we can use exact samplers; for generalized Laplace we develop an efficient and practical solution using Markov chain Monte Carlo (MCMC). Our experiments all use the MCMC sampler for practicality.

{\bf Sampling discrete Gaussian on a lattice.}
There are several efficient algorithms to sample discrete variables within negligible statistical distance of any discrete Gaussian distribution whose scales are not exceedingly narrow~\citep{gentry2008trapdoors,peikert2010efficient}.  Moreover, \citep{DBLP:journals/corr/BrakerskiLPRS13} provide an exact Gaussian sampler.  
\begin{theorem}[name =  Lemma 2.3 in \citet{DBLP:journals/corr/BrakerskiLPRS13}] There is a probabilistic polynomial-time algorithm that takes as input a basis $\mB = (\vb_1, \dots, \vb_m)$ for a lattice $\Lambda(\mB)\subset \R^d$, a center $\vc\in \R^d$ and parameter $\sigma = \Omega(\max_i\|\vb_i\|_2\ln d)$ and
outputs a vector that is distributed exactly as \cref{eq:gen-gaussian}.
\end{theorem}
We note that, however,  it is believed to be computationally difficult to sample efficiently on general lattice spaces if the target distribution is not sufficiently dispersed. 
Most lattice-based cryptographic schemes~\citep{follath2014gaussian} are based on the hardness assumption of solving the short integer solution (SIS) problem and the learning with errors (LWE) problem~\citep{DBLP:journals/corr/BrakerskiLPRS13}, and the ability of an efficient sampler of a narrow distribution on general lattice efficiently would enable us to solve short integer solution and break lattice-based cryptographic schemes. 

{\bf MCMC sampling for generalized Laplace on a lattice.}
Since a discrete Gaussian mechanism only provides $(\eps, \delta)$-DP, it is important to develop pure-DP mechanisms using the generalized Laplace mechanism on a lattice. 
We devise an MCMC sampling scheme to instantiate the generalized Laplace mechanism (\cref{def:lap}).  
As a practical solution, MCMC has been theoretically studied in the literature; e.g.~\citet{Wang2015-dz} for stochastic gradient Monte Carlo and~\citet{Ganesh_undated-ty} for discretized Langevin MCMC.
The challenge to designing such a sampling scheme is twofold. First, the invariants render the target state space highly constrained. To achieve sampling efficiency requires efficient proposals, which ideally satisfy the invariants themselves in order to maintain a reasonable acceptance rate \cite[c.f.][]{gong2020congenial}. Second, MCMC algorithms are generally not exact samplers, and may fail to converge to the target distribution in finite time, thus incur additional cost to privacy that can be difficult to quantify for a given application. To this end, we 1) target only the additive noise component of the mechanism, ensuring that any privacy leakage due to sampling is independent of the underlying confidential data; and 2) devise an empirical convergence assessment for the upper bound on the total variation (TV) distance between the chain's marginal distribution the target, based on the $L$-lag coupling method proposed by~\citep{biswas2019estimating}.

{\it Proposal design.} We wish to sample the additive noise $\vz$ employed by the generalized Laplace mechanism. 
For an incidence matrix $\mA\in \Z^{k\times d}$ specifying the invariants, recall its Smith normal form $\mD = \mU\mA\mV$  in~\cref{eq:smith-normal-form}. Consider a jumping distribution with probability mass function $g\left(\vu\right)=g_{\mA}\left(\mV^{-1}\vu\right),$ where
\begin{equation}\label{eq:proposal}
    g_{\mA}\left(\ve \right)=\prod_{j=1}^{k}{\bf 1}\left\{ e_{j}=0\right\} \prod_{j=k+1}^{d}\eta_{j}\left(e_{j}\right),
\end{equation}
for $\ve = \left(e_{1},\ldots,e_{d}\right)$ a \emph{pre-jump} object
whose first $k$ entries are exactly zero, and $\vu =  \mV\ve$ the proposed jump.
The $\eta_{j}$'s in~\cref{eq:proposal} can be set to any interger-valued univariate probabilities that are unbiased and symmetric around zero. A straightforward choice is the double geometric distribution, with mass function
$\eta_j\left(e\right)=\frac{1-a}{1+a}a^{\left\Vert e\right\Vert_{1}}$ and parameter $a$.
Note that the jumping distribution $g$ connects to $g_\mA$ in~\cref{eq:proposal} without a Jacobian because $\vu$ and $\ve$ are integer-valued and $\mV$ is invertible. The proposed jump $\vu \sim g$ is unbiased and symmetric, i.e. $\mathbb{E}_{g}(\vu) = {\bf 0}$ and $g(\vu) = g(-\vu)$, and always respects the desired invariants specified by $\mA$: 
\[
\mA\vu=\mU^{-1}\mU\mA\mV\ve=\mU^{-1}\mD\ve=\mU^{-1}{\bf 0}={\bf 0}.
\]

{\it A Metropolis sampling scheme.}
 The target distribution $q_{\epsilon}\left(\vz\right)$ is given in~\cref{eq:gen-laplace} and is known only up to a normalizing constant. ~\cref{alg:metropolis-gibbs} in ~\cref{appendix:mcmc} presents a Gibbs-within-Metropolis sampler that produces a sequences of dependent draws  $\left(\vz^{\left(l\right)}\right)_{0\le l\le \text{nsim}}$ from the target distribution $q_{\epsilon}$ in~\cref{eq:gen-laplace} known only up to a normalizing constant. We use an additive jumping distribution whose element-wise construction is described in~\cref{eq:proposal}.  
The algorithm incurs a transition kernel that dictates how the chain moves from an existing state to the next one: $\vz^{\left(l\right)} \sim K\left(\vz^{\left(l-1\right)},\cdot\right)$. The initial distribution $\pi_0$ may be chosen simply as $g$ for convenience. The choice of $\eta_j$ is a tuning decision for the algorithm,
and should be made to encourage fast mixing of the chain. We discuss our choices for examples in the Experiments section. Note that steps 5 through 7 of ~\cref{alg:metropolis-gibbs} updates the proposed jump in a Gibbs sweep (i.e. one dimension at a time), utilizing the fact that the pre-jump object $\ve$ is element-wise independent under $g_\mA$. Doing so facilitates the $L$-lag coupling, to be discussed next, for assessing empirical convergence of the chain. In practice, these updates may be performed simultaneously.

{\it Empirical convergence assessment using $L$-lag coupling.}
We perform empirical convergence assessment of the proposed algorithm using $L$-lag coupling.
Construct a joint transition kernel $\tilde{K}$ of two Markov chains, each having the same target distribution $q_\epsilon$ induced by the marginal transition kernel $K$ as defined in ~\cref{alg:metropolis-gibbs}. 
The joint kernel $\tilde{K}$, given by ~\cref{alg:coupling-gibbs} in ~\cref{appendix:mcmc}, is a maximal coupling kernel such that the $L$-th lag of the two chains will couple in finite time with probability one. The random $L$-lag meeting time, which ~\cref{alg:llag-alt} in ~\cref{appendix:mcmc} samples, provides an estimate of the upper bound on the TV distance, $$d_{TV}(q_{\epsilon}^{(l)},q_{\epsilon})\le\mathbb{E}\left[\max\left(0,\left\lceil \left(\tau^{(L)}-L-l\right)/L\right\rceil \right)\right]$$
between  the target distribution  $q_{\epsilon}$ and the marginal distribution of the chain at time $l$, denoted $q_{\epsilon}^{(l)}$ \citep[Theorem 2.5]{biswas2019estimating}. The upper bounds are obtained as empirical averages over independent runs of coupled Markov chains. The lag $L>0$ would need to be set, and we discuss its choice in the next section and ~\cref{appendix:mcmc}. 
Compared to earlier work which uses MCMC to instantiate privacy mechanisms, our use of $L$-lag coupling provides real-time (rather than asymptotic) assessment on the MCMC convergence behavior. The number of iterations needed to ensure convergence is empirically assessed via the upper bound on the total variation (note that this is still an estimate). In general, there is a tradeoff between the number of iterations and an extra privacy loss budget $\delta'$, in the sense that Eq. \cref{eq:gen-laplace} can be absorbed as another additive error in the DP guarantee.

\section{Experiments}\label{sec:experiments}

\begin{figure}
    \centering
    \includegraphics[width=0.5\textwidth]{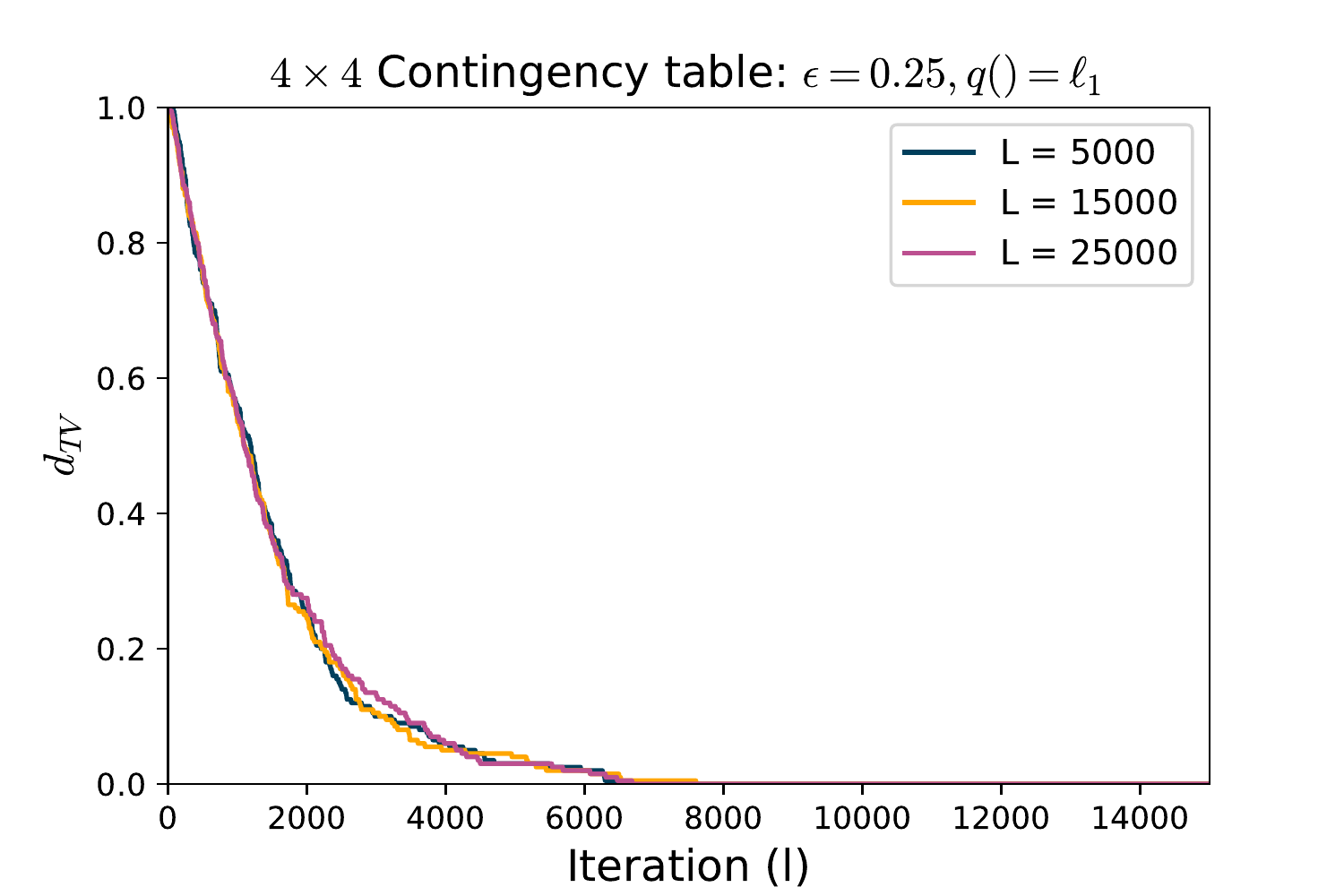}
    \captionof{figure}{Estimated TV bound from target\label{fig:children-coupling-l1}}
\end{figure}

\begin{table}[]
    \centering
    \begin{small}
\begin{tabular}{cccccc}
     \toprule
     County & Low & Medium & High & Very High & Total \\
     \midrule
     Alpha & -1 & -1 & 0 & 2 & 0 \\
     Beta & 5 & 4 & -2 & -7 & 0 \\
     Gamma & -5 &-3 & 5 & 3 & 0 \\
     Delta & 1 & 0 & -3 & 2 & 0 \\
     \midrule
     Total & 0 & 0 & 0 & 0 & 0 \\
     \bottomrule
\end{tabular}
    \caption{Generalized Laplace additive noise  ($\epsilon=0.25$, $\ell_1$-norm target distribution)  using ~\cref{alg:metropolis-gibbs}, which preserves row and column margins.}
    \label{table:children-noise}
    \vspace*{-4mm}
    \end{small}
\end{table}

\begin{figure*}[t]
    \centering
    \includegraphics[width=.95\textwidth]{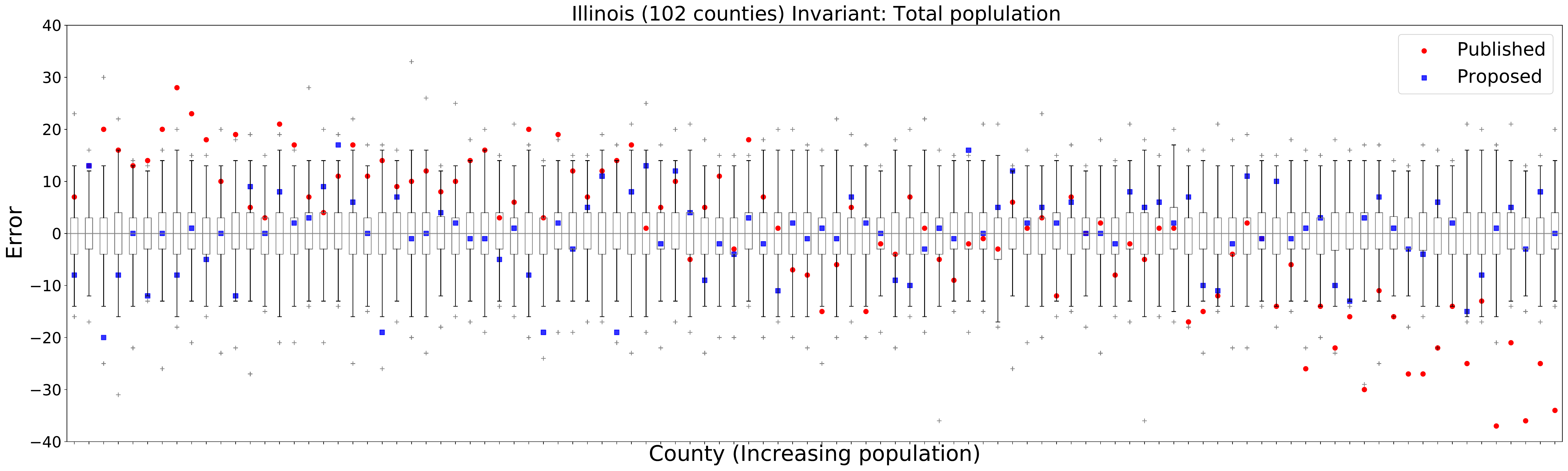}
    \caption{Privacy noise from the generalized Laplace mechanism (~\cref{def:lap}) via ~\cref{alg:metropolis-gibbs} (blue squares: one instance; boxplot: 1000 instances) for county populations of Illinois in increasing county sizes. State population total is invariant. The proposed noises are integer-valued and unbiased.
    For comparison are DAS errors from the Nov 2020 vintage 2010 Census demonstration data \citep{ipums2020das} (red dots). 
    }
    \label{fig:illinois}
    \vspace*{-4mm}
\end{figure*}

{\bf Intersecting counting constraints.} Consider a synthetic example in which a set of three counting constraints $\mathcal{A} = (A_1, A_2, A_3)$ are defined over 14 records, where the intersection of all subsets of constrains are nonempty. The constraints are schematically depicted by ~\cref{fig:circles-constraint} in ~\cref{appendix:experiments}, and may be encoded by an incident matrix $\mA \in \Z^{3\times 14}$, with each row corresponding to one of $A_1, A_2$ and $A_3$, and each element being $1$ at indices corresponding to the record within that constraint and $0$ otherwise.  We apply ~\cref{alg:metropolis-gibbs} to instantiate the generalized Laplace mechanism, with both $\ell_1$-norm and $\ell_2$-norm targets, to privatize a data product that conforms to the constraint $\mA$.   To ensure adequate dispersion of the target distribution, we set $\epsilon = 0.25$, a value on the smaller end within the range of meaningful privacy protection \cite[e.g.][]{dwork2011firm}.
The pre-jump proposal distributions $\eta_j$ are double geometric distributions, with parameter $a=\exp{(-1)}$ for the $\ell_1$-norm target and $a=\exp{(-1.5)}$ for the $\ell_2$-norm target.  Details on the tuning, the noise distribution and convergence assessment can be found in ~\cref{appendix:experiments}. In particular, ~\cref{fig:circles-coupling} in ~\cref{appendix:experiments}  shows that the chains empirically converge at around $10^5$ iterations for the $\ell_1$-norm target and $10^6$ iterations for the $\ell_2$-norm target. 


\noindent{\bf Delinquent children by county and household head education level.} The Federal Committee on Statistical Methodology published a fictitious dataset concerning delinquent children in the form of a $4\times 4$ contingency table, tabulated across four counties by education level of household head \cite[Table 4 in ][reproduced in ~\cref{table:children} of ~\cref{appendix:experiments}]{fcsm2005sdl}, to illustrate various traditional SDL techniques \citep{slavkovic2010synthetic}.


The charge is to extend privacy protection to the sensitive individual records while preserving the margins of the contingency table for data publication. To do so, we apply ~\cref{alg:metropolis-gibbs} to instantiate the generalized Laplace mechanism with both $\ell_1$-  and $\ell_2$-norm targets and with $\epsilon = 0.25$. 
~\cref{fig:children-coupling-l1} shows the evolution of the TV upper bound on the chain's marginal distributions to the $\ell_1$-norm target, estimated with 200 independent coupled chains, which appear to converge after about $10^4$ iterations and are stable at various choices of $L$. ~\cref{table:children-noise} shows one instance of the proposed additive noise, obtained after the chain achieves empirical convergence. They are integer-valued, with zero row and column totals which would preserve the margins of the confidential table. Convergence assessment for the $\ell_2$-norm target and discussions on the choice of proposal are in ~\cref{appendix:experiments}.

{\bf 2010 U.S. Census county-level population data.} We consider the publication of county-level population counts subject to the invariant of state population size, 
and compare with the privacy-protected demonstration files produced by preliminary versions of the 2020 Census DAS. The confidential values are the 2010 Census Summary Files (CSF), curated by IPUMS NHGIS and are publicly available \citep{ipums2020das}. Employed for this example are the November 2020 vintage demonstration data, protected by pure differential privacy with $\epsilon = 0.192 = 4 \text{ (total)} \times 0.16 \text{ (county level)} \times 0.3 \text{ (population query)}$.

We demonstrate ~\cref{alg:coupling-gibbs} using the generalized Laplace mechanism under $\ell_1$ norm, with $\epsilon$ set to accord to the Census Bureau's specification and $a=\exp(-2.5)$. ~\cref{fig:illinois} showcases the proposed county-level errors
applied to the population of Illinois. The $x$-axis is arranged in increasing true county population sizes.  Blue squares shows an instance of the proposed noise vector. The boxplots summarize $1000$ proposed noise vectors (thinned at $0.01\%$), with whiskers indicating $1.5$ times the interquartile range and cross-marks indicating extreme realizations.  The proposed errors are integer-valued and centered around zero, confirming their marginal unbiasedness. Note that the published DAS privacy errors (red dots) show a clear negative bias as a function of increasing underlying county population size. The bias is likely due to optimization-based post-processing during the estimation phase which imposes non-negativity on the constituent geographic areas with small population counts. Similar observations were made for other states (\cref{appendix:experiments}).

\section{Discussion and Future Work}\label{sec:discussion}

This paper provides solutions for sanitizing private data that are required to simultaneously observe pre-specified invariants and integral characteristics. The proposed \emph{integer subspace differential privacy} scheme allows for rigorous statements of privacy guarantees while maintaining good statistical properties, including unbiasedness, accuracy, and probabilistic transparency of the privatized output. An efficient MCMC scheme is devised to instantiate the proposed mechanism, alongside tools to assess empirical convergence.

One direction for future work is to design exact sampling algorithms for the proposed generalized discrete Laplace and Gaussian mechanisms. 
In the case that the normalizing constant to the target distribution $q_{\epsilon}$ is known exactly, it may be feasible to design, for example, efficient rejection sampling based on intelligent choices of proposal distributions. It is also conceivable that a perfect sampling scheme, such as \emph{coupling from the past} \citep{propp1996exact}, may be designed. The recent work of \citet{seeman2021exact} proposes an exact sampling scheme for Exponential mechanisms using \emph{artificial atoms}. However, the technique requires the state space be compact, which is not the case for the constrained yet unbounded target distribution $q_{\epsilon}$ considered in this work.
Another direction for future work is to extend invariant-respecting privacy protection data products that are required to be \emph{binary} (i.e., exhibiting or not exhibiting a private attribute), as well as to obey \emph{inequality} constraints (e.g. non-negativity of count data). The challenge there is again on the design of efficient methods to generate noise vectors in an extremely constrained discrete space.


\section{Acknowledgements} 
Prathamesh Dharangutte and Jie Gao would like to acknowledge support from NSF through grants CCF-2118953, IIS-2207440, CCF-2208663 and CNS-2137245. Ruobin Gong would like to acknowledge support from NSF through grant DMS-1916002.


\bibliographystyle{apalike}
\bibliography{main_arxiv}

\begin{thebibliography}{}

\bibitem[Abowd et~al., 2022]{abowd2022topdown}
Abowd, J.~M., Ashmead, R., Cumings-Menon, R., Garfinkel, S., Heineck, M.,
  Heiss, C., Johns, R., Kifer, D., Leclerc, P., Machanavajjhala, A., Moran, B.,
  Sexton, W., Spence, M., and Zhuravlev, P. (2022).
\newblock The census topdown algorithm.
\newblock {\em Harvard Data Science Review}.

\bibitem[Ashmead et~al., 2019]{ashmead2019effective}
Ashmead, R., Kifer, D., Leclerc, P., Machanavajjhala, A., and Sexton, W.
  (2019).
\newblock Effective privacy after adjusting for invariants with applications to
  the 2020 census.
\newblock Technical report, US Census Bureau.

\bibitem[Barak et~al., 2007]{barak2007privacy}
Barak, B., Chaudhuri, K., Dwork, C., Kale, S., McSherry, F., and Talwar, K.
  (2007).
\newblock Privacy, accuracy, and consistency too: a holistic solution to
  contingency table release.
\newblock In {\em Proceedings of the twenty-sixth ACM SIGMOD-SIGACT-SIGART
  symposium on Principles of database systems}, pages 273--282.

\bibitem[Barber and Duchi, 2014]{https://doi.org/10.48550/arxiv.1412.4451}
Barber, R.~F. and Duchi, J.~C. (2014).
\newblock Privacy and statistical risk: Formalisms and minimax bounds.
\newblock arXiv 1412.4451.

\bibitem[Bentkus and Gotze, 1999]{bentkus1999lattice}
Bentkus, V. and Gotze, F. (1999).
\newblock Lattice point problems and distribution of values of quadratic forms.
\newblock {\em Annals of Mathematics}, 150(3):977--1027.

\bibitem[Biswas et~al., 2019]{biswas2019estimating}
Biswas, N., Jacob, P.~E., and Vanetti, P. (2019).
\newblock Estimating convergence of markov chains with l-lag couplings.
\newblock {\em Advances in Neural Information Processing Systems}, 32.

\bibitem[Brakerski et~al., 2013]{DBLP:journals/corr/BrakerskiLPRS13}
Brakerski, Z., Langlois, A., Peikert, C., Regev, O., and Stehl\'{e}, D. (2013).
\newblock Classical hardness of learning with errors.
\newblock page 575–584.

\bibitem[Canonne et~al., 2020]{canonne2020discrete}
Canonne, C.~L., Kamath, G., and Steinke, T. (2020).
\newblock The discrete gaussian for differential privacy.
\newblock {\em Advances in Neural Information Processing Systems},
  33:15676--15688.

\bibitem[Chatzikokolakis et~al., 2013]{chatzikokolakis2013broadening}
Chatzikokolakis, K., Andr{\'{e}}s, M.~E., Bordenabe, N.~E., and Palamidessi, C.
  (2013).
\newblock Broadening the scope of differential privacy using metrics.
\newblock In Cristofaro, E.~D. and Wright, M.~K., editors, {\em Privacy
  Enhancing Technologies - 13th International Symposium, {PETS} 2013,
  Bloomington, IN, USA, July 10-12, 2013. Proceedings}, volume 7981 of {\em
  Lecture Notes in Computer Science}, pages 82--102. Springer.

\bibitem[Cormode et~al., 2017]{cormode2017constrained}
Cormode, G., Kulkarni, T., and Srivastava, D. (2017).
\newblock Constrained differential privacy for count data.
\newblock {\em arXiv preprint arXiv:1710.00608}.

\bibitem[Desfontaines and Pej{\'{o}}, 2019]{damien2019sok}
Desfontaines, D. and Pej{\'{o}}, B. (2019).
\newblock Sok: Differential privacies.
\newblock {\em CoRR}, abs/1906.01337.

\bibitem[Dwork, 2011]{dwork2011firm}
Dwork, C. (2011).
\newblock A firm foundation for private data analysis.
\newblock {\em Communications of the ACM}, 54(1):86--95.

\bibitem[Dwork et~al., 2006]{dwork2006calibrating}
Dwork, C., McSherry, F., Nissim, K., and Smith, A. (2006).
\newblock Calibrating noise to sensitivity in private data analysis.
\newblock In {\em Theory of cryptography conference}, pages 265--284. Springer.

\bibitem[{Federal Committee on Statistical Methodology}, 2005]{fcsm2005sdl}
{Federal Committee on Statistical Methodology} (2005).
\newblock Report on statistical disclosure limitation methodology.
\newblock Technical report, Statistical Policy Working Paper 22 (Second
  Version).
\newblock \url{https://www.hhs.gov/sites/default/files/spwp22.pdf} [Accessed:
  05-15-2022].

\bibitem[Foll\'{a}th, 2014]{follath2014gaussian}
Foll\'{a}th, J. (2014).
\newblock Gaussian sampling in lattice based cryptography.
\newblock {\em Tatra Mountains Mathematical Publications}, 60(1):1--23.

\bibitem[{Ganesh} and {Talwar}, 2020]{Ganesh_undated-ty}
{Ganesh} and {Talwar} (2020).
\newblock Faster differentially private samplers via r{\'e}nyi divergence
  analysis of discretized langevin {MCMC}.
\newblock In {\em Adv. Neural Inf. Process. Syst.}, pages 7222--7233.

\bibitem[Gao et~al., 2022]{gao2022subspace}
Gao, J., Gong, R., and Yu, F.-Y. (2022).
\newblock Subspace differential privacy.
\newblock In {\em Proceedings of the The Thirty-Sixth AAAI Conference on
  Artificial Intelligence (AAAI-22)}, pages 3986--3995.
\newblock arXiv:2108.11527.

\bibitem[Gentry et~al., 2008]{gentry2008trapdoors}
Gentry, C., Peikert, C., and Vaikuntanathan, V. (2008).
\newblock Trapdoors for hard lattices and new cryptographic constructions.
\newblock In {\em Proceedings of the fortieth annual ACM symposium on Theory of
  computing}, pages 197--206.

\bibitem[Gilbert and Pathria, 1990]{gilbert1990linear}
Gilbert, W.~J. and Pathria, A. (1990).
\newblock Linear diophantine equations.
\newblock {\em preprint}.

\bibitem[Gong, 2022]{gong2022transparent}
Gong, R. (2022).
\newblock Transparent privacy is principled privacy.
\newblock {\em Harvard Data Science Review (Special Issue 2)}.
\newblock doi:10.1162/99608f92.b5d3faaa.

\bibitem[Gong and Meng, 2020]{gong2020congenial}
Gong, R. and Meng, X.-L. (2020).
\newblock Congenial differential privacy under mandated disclosure.
\newblock In {\em Proceedings of the 2020 ACM-IMS on Foundations of Data
  Science Conference}, pages 59--70.

\bibitem[Greenberg, 1971]{greenberg1971integer}
Greenberg, H. (1971).
\newblock {\em Integer programming}.
\newblock Academic Press.

\bibitem[Hay et~al., 2009]{hay2009boosting}
Hay, M., Rastogi, V., Miklau, G., and Suciu, D. (2009).
\newblock Boosting the accuracy of differentially-private histograms through
  consistency.
\newblock {\em arXiv preprint arXiv:0904.0942}.

\bibitem[He et~al., 2014]{he2014blowfish}
He, X., Machanavajjhala, A., and Ding, B. (2014).
\newblock Blowfish privacy: Tuning privacy-utility trade-offs using policies.
\newblock In {\em Proceedings of the 2014 ACM SIGMOD international conference
  on Management of data}, pages 1447--1458.

\bibitem[Hotz and Salvo, 2020]{hotz2020assessing}
Hotz, J.~V. and Salvo, J. (2020).
\newblock Assessing the use of differential privacy for the 2020 {Census}:
  Summary of what we learned from the {CNSTAT} workshop.
\newblock Technical report, National Academies Committee on National Statistics
  (CNSTAT).
\newblock
  \url{https://www.amstat.org/asa/files/pdfs/POL-CNSTAT_CensusDP_WorkshopLessonsLearnedSummary.pdf}
  [Accessed: 04-08-2020].

\bibitem[Kifer and Machanavajjhala, 2014]{kifer2014pufferfish}
Kifer, D. and Machanavajjhala, A. (2014).
\newblock Pufferfish: A framework for mathematical privacy definitions.
\newblock {\em ACM Transactions on Database Systems (TODS)}, 39(1):1--36.

\bibitem[Landau, 1915]{landau1915analytischen}
Landau, E. (1915).
\newblock Zur analytischen zahlentheorie der definiten quadratischen formen
  (uber die gitterpunkte in einem mehrdimensionalen ellipsoid), sitzungsber.
  koniglich breissischen.
\newblock {\em Akad. Wiss.}, 31:458--476.

\bibitem[Peikert, 2010]{peikert2010efficient}
Peikert, C. (2010).
\newblock An efficient and parallel gaussian sampler for lattices.
\newblock In {\em Annual Cryptology Conference}, pages 80--97. Springer.

\bibitem[Propp and Wilson, 1996]{propp1996exact}
Propp, J.~G. and Wilson, D.~B. (1996).
\newblock Exact sampling with coupled markov chains and applications to
  statistical mechanics.
\newblock {\em Random Structures \& Algorithms}, 9(1-2):223--252.

\bibitem[Schrijver, 1998]{schrijver1998theory}
Schrijver, A. (1998).
\newblock {\em Theory of linear and integer programming}.
\newblock John Wiley \& Sons.

\bibitem[Seeman et~al., 2021]{seeman2021exact}
Seeman, J., Reimherr, M., and Slavkovi{\'c}, A. (2021).
\newblock Exact privacy guarantees for markov chain implementations of the
  exponential mechanism with artificial atoms.
\newblock {\em Advances in Neural Information Processing Systems}, 34.

\bibitem[Seeman et~al., 2022]{seeman2022partially}
Seeman, J., Slavkovic, A., and Reimherr, M. (2022).
\newblock A formal privacy framework for partially private data.
\newblock {\em arXiv preprint arXiv:2204.01102}.

\bibitem[Slavkovi{\'c} and Lee, 2010]{slavkovic2010synthetic}
Slavkovi{\'c}, A.~B. and Lee, J. (2010).
\newblock Synthetic two-way contingency tables that preserve conditional
  frequencies.
\newblock {\em Statistical Methodology}, 7(3):225--239.

\bibitem[Tsuji, 1953]{tsuji1953lattice}
Tsuji, M. (1953).
\newblock On lattice points in an n-dimensional ellipsoid.
\newblock {\em Journal of the Mathematical Society of Japan}, 5(3-4):295--306.

\bibitem[{U.S. Census Bureau}, 2019]{census2019memo}
{U.S. Census Bureau} (2019).
\newblock Memorandum 2019.25: 2010 demonstration data products - design
  parameters and global privacy-loss budget.
\newblock
  \url{https://www.census.gov/programs-surveys/decennial-census/decade/2020/planning-management/plan/memo-series/2020-memo-2019_25.html}
  [Accessed: 04-18-2022].

\bibitem[{U.S. Census Bureau}, 2021]{census2021PLB}
{U.S. Census Bureau} (2021).
\newblock Privacy-loss budget allocation 2021-06-08.
\newblock
  \url{https://www2.census.gov/programs-surveys/decennial/2020/program-management/data-product-planning/2010-demonstration-data-products/ppmf20210608/2021-06-08-privacy-loss_budgetallocation.pdf}
  [Accessed: 02-15-2022].

\bibitem[Van~Riper et~al., 2020]{ipums2020das}
Van~Riper, D., Kugler, T., and Schroeder, J. (2020).
\newblock {IPUMS NHGIS} privacy-protected 2010 {Census} demonstration data,
  version 20201116.
\newblock Minneapolis, MN: IPUMS.

\bibitem[Wang et~al., 2015]{Wang2015-dz}
Wang, Y.-X., Fienberg, S.~E., and Smola, A. (2015).
\newblock Privacy for free: Posterior sampling and stochastic gradient monte
  carlo.
\newblock In {\em Proceedings of the International Conference on Machine
  Learning}, pages 2493--2502.

\bibitem[Zhu et~al., 2021]{zhu2021bias}
Zhu, K., Van~Hentenryck, P., and Fioretto, F. (2021).
\newblock Bias and variance of post-processing in differential privacy.
\newblock In {\em Proceedings of the AAAI Conference on Artificial
  Intelligence}, volume~35, pages 11177--11184.

\end{thebibliography}
\clearpage

\newpage

\appendix

\section{Proofs of Composition and Post-Processing}\label{appendix:composition-post-processing}

\begin{proof}[Proof of \cref{prop:composition}]

Let $c_1 = \frac{e^{\epsilon_1\|\vx-\vx'\|}-1}{e^{\epsilon_1}-1}\delta_1$, $c_2 = \frac{e^{\epsilon_2\|\vx-\vx'\|}-1}{e^{\epsilon_2}-1}\delta_2$ and $x \wedge y:= \min(x,y)$. For any fixed output $(y_1, y_2)$ we have
\begin{equation*}
    \begin{split}
        &\Pr[M_{1,2}(\vx) = (y_1, y_2)] \\
        =& \Pr[M_1(\vx) = y_1] \Pr[M_2(\vx) = y_2 | M_1(\vx) = y_1] \\
\leq & \Pr[M_1(\vx) = y_1]\\
& \hspace{20pt} \cdot(e^{\epsilon_2\|\vx-\vx'\|}\Pr[M_2(\vx') = y_2 | M_1(\vx) = y_1] \wedge 1 + c_2) \\
\leq & \big( e^{\epsilon_2\|\vx-\vx'\|}\Pr[M_2(\vx') = y_2 | M_1(\vx) = y_1] \wedge 1 \big)  \\ 
& \hspace{20pt} \cdot(e^{\epsilon_1\|\vx-\vx'\|}\Pr[M_1(\vx') = y_1] + c_1) + c_2 \\
\leq & e^{(\epsilon_1+\epsilon_2)\|\vx-\vx'\|} \Pr[M_{1,2}(\vx') = (y_1,y_2)] + (c_1 + c_2)
    \end{split}
\end{equation*}
which completes the proof.
\end{proof}

\begin{proof}[Proof of \cref{prop:post-processing}]
For $M: \mathcal{X} \to \mathcal{Y}$ which is $(\epsilon,\delta)$-differentially private on $\equiv$, given a output space $\mathcal{Z}$, and an stochastic mapping $F$ from $\mathcal{Y}$ to $\mathcal{Z}$,  for any event $E\subset \mathcal{Z}$, we have
\begin{equation*}
    \begin{split}
        & \Pr[F(M(\vx))\in E] \\
        =& \E_{f\sim F}[\Pr[f(M(\vx))\in E]] \\
        =& \E_{f\sim F}[\Pr[M(\vx)\in f^{-1}(E)]] \\
        \leq &\E_{f\sim F}\Big[ e^{\epsilon \|\vx-\vx'\|} \Pr[M(\vx') \in f^{-1}(E)] + \frac{e^{\epsilon\|\vx-\vx'\|}-1}{e^{\epsilon}-1}\delta \Big] \\
        = & e^{\epsilon \|\vx-\vx'\|} \Pr[F(M(\vx'))\in E] +\frac{e^{\epsilon\|\vx-\vx'\|}-1}{e^{\epsilon}-1}\delta
    \end{split}
\end{equation*}
where the second last inequality follows form \cref{def:partitional_dp}.
\end{proof}

\section{Proofs for Laplace mechanisms}\label{appendix:lattice}

\begin{proof}[Proof of \cref{thm:lap}]
By definition, the mechanism is $\A$-invariant. 
Now we show the mechanism is $\epsilon$-$\A$-induced integer subspace differentially private on $\A$. For all $\vx\equiv_\A \vx'$ and a feasible outcome $\vy$ so that $\vy\equiv_\A\vx$, by \cref{prop:feasible} $\vy\in \vx+\Lambda_\A = \vx'+\Lambda_\A$.  The privacy loss is 
\begin{align*}
    \ln \frac{\Pr[M(\vx) = \vy]}{\Pr[M(\vx') = \vy]} =& \ln \frac{q_\epsilon(\|\vy-\vx\|)}{q_\epsilon(\|\vy-\vx'\|)}\\
    & \cdot \ln \frac{\sum_{\vw\in \vx+\Lambda_\A}q_\epsilon(\|\vw-\vx\|)}{\sum_{\vw\in\vx'+\Lambda_\A}q_\epsilon(\|\vw-\vx'\|)}
\end{align*}
For the first term, $$\ln \frac{q_\epsilon(\|\vy-\vx\|)}{q_\epsilon(\|\vw-\vx'\|)} = -\epsilon\left(\|\vy-\vx\|-\|\vy-\vx'\|\right)\le \epsilon \|\vx-\vx'\|,$$ by the triangle inequality of $\|\cdot\|$.  For the second term, $\vw$ is feasible if and only if $\vw+\vx-\vx'$ is still feasible because $\vx\equiv_\A \vx'$.  Therefore
\begin{align*}
\sum_{\vw\in \vx+\Lambda_\A}q_\epsilon(\|\vw-\vx\|) &= \sum_{\vw+\vx-\vx'\in \vx+\Lambda_\A}q_\epsilon(\|\vw+\vx-\vx'-\vx\|)\\
    & = \sum_{\vw \in \vx'+\Lambda_\A}q_\epsilon(\|\vw-\vx'\|)
\end{align*}
and the second term is zero.  Combining these two we have
$$\ln \frac{\Pr[M(\vx) = \vy]}{\Pr[M(\vx') = \vy]} \le \epsilon \|\vx-\vx'\|$$
which completes the proof.
\end{proof}

Now we show the unbiasedness and accuracy guarantee in \cref{thm:lap_acc}.  To control the error's tail bound, we first show that lattice spaces are similar to Euclidean spaces which are ``regular'' and have bounded dimensions.   We use the idea of doubling dimension that bounds the grow rate of the number of lattice points in spheres as the radius increases in~\cref{lem:doubling}  Formally, given a positive definite matrix $\mQ\in \R^{l\times l}$ and $s>0$, we define the ellipsoid as a set $$E_\mQ(s) = \{\vv\in \R^l:\vv^\top \mQ \vv\le s^2\}.$$
We use $\vol (E_\mQ(s))$ as the Lebesgue measure of $E_\mQ(s)$ and $\vol_\Z(E_\mQ(s)) := |\Z^l\cap E_\mQ(s)|$ as the number of points in $E_\mQ(s)$ with integer coordinates.  There is a long line of work dating back to \citep{landau1915analytischen,tsuji1953lattice,bentkus1999lattice} showing that the number of lattice points in $E_\mQ(s)$ is approximated by the volume of $E_\mQ(s)$, when $s$ is large enough.

\begin{theorem}\citep{landau1915analytischen}
For all $\ell\ge 2$ and $\mQ$, when $s>0$ is large enough, we have
$$\frac{|\vol (E_\mQ(s))-\vol_\Z(E_\mQ(s))|}{\vol (E_\mQ(s))} = \mathcal{O}\left(s^{-2+\frac{2}{1+\ell}}
\right)$$
\end{theorem} 
Additionally, let the volume of a unit $l$-dimensional sphere be $V_l = \vol(\mathbb{I}) = \frac{\pi^{l/2}}{\Gamma(l/2+1)}$ where $\Gamma$ is the gamma function.  The volume of a general ellipsoid $E_\mQ(s)$ of dimension $l$ is $V_l\left(\frac{s^2}{\det(\mQ)}\right)^{l/2}$ where $\det(\mQ)$ is the determinant of $\mQ$.  Therefore, we have for all positive definite $\mQ$ and $\rho$ there exists $s_\rho$ so that for all $s\ge s_\rho$
\begin{equation}\label{eq:dimension}
    \frac{1}{\rho} s^{l}\le \frac{\vol_\Z(E_\mQ(s))}{\frac{V_l}{\sqrt{\det(\mQ)}}}\le \rho s^{l}.
\end{equation}
Informally, as the radius of the ellipsoid doubles, the number of integer points in the ellipsoid increases by a factor $2^l$, where $l$ is the dimension.  Now we are ready to prove \cref{lem:doubling} and \cref{thm:lap_acc}.

\begin{proof}[Proof of \cref{lem:doubling}]
Observe that 
\begin{align*}
    |\{\vz\in \Lambda(\mB): \|\vz\|\le r\}| =&|\{\vv\in \Z^{m}: \|\mB\vv\|\le r\}| \\
    =&|\{\vv\in \Z^{m}: \vv^\top (\mB^\top\mB)\vv\le r^2\}|\\
    =&\vol_\Z\left(E_{\mB^\top\mB}(r)\right)
\end{align*}
Because $\mB^\top\mB$ is full rank, by \cref{eq:dimension}, there exists $r^*>0$ so that for all $r\ge r^*$, 
$$ |B_\A(r)|\le \frac{2V_{m}}{\sqrt{\det(\mB^\top\mB)}}r^{m},$$
which completes the proof.
\end{proof}

\begin{proof}[Proof of \cref{thm:lap_acc}]

Let $\rvz = M_{Lap,\A,\epsilon}(\vx)-\vx$ in \cref{def:lap} be a random vector.

First it is easy to show that the noise of the generalized Laplace mechanism is unbiased, because the density $\exp\left(-\epsilon\|\vz\|\right)\mathbf{1}[\vz\in \Lambda_\A]$ is an even function, for all $\vz\in \Z^d$ $$\exp\left(-\epsilon\|-\vz\|\right)\mathbf{1}[-\vz\in \Lambda_\A] = \exp\left(-\epsilon\|\vz\|\right)\mathbf{1}[\vz\in \Lambda_\A].$$

Given $\Lambda_\A \subseteq \R^d$, let $D_\A(r) = \{\vz\in \Lambda_\A: \|\vz\|\le r\}$.  The tail probability of $\vz$ is
$$\Pr[\|\vz\|\ge t] = \frac{\sum_{\vz\in \Lambda_\A\setminus D_\A(t)} \exp(-\epsilon \|\vz\|)}{\sum_{\vz\in \Lambda_\A} \exp(-\epsilon \|\vz\|)}.$$
We can upper bound it by lower bounding the denominator and upper bounding the numerator. 

For the denominator, 
$$\sum_{\vz\in \Lambda_\A} \exp(-\epsilon \|\vz\|)\ge \exp(-\epsilon \|\mathbf{0}\|) = 1.$$
For the numerator, by \cref{lem:doubling}, because $\mC_\A$ has full rank, there exists a constant $L= \frac{2V_{d-k}}{\sqrt{\det(\mC_\A^\top\mC_\A)}}$, 
\begin{equation}\label{eq:doubling2}
    |D_\A(r)|\le L r^{d-k}, \text{ for all large enough } r>0.
\end{equation}
Hence, if $t$ is large enough, we have
\begin{align*}
    &\sum_{\vz\in \Lambda_\A\setminus D_\A(t)} \exp(-\epsilon \|\vz\|)\\
    = & \sum_{l=0}^\infty \sum_{\vz\in D_\A(2^{l+1}t)\setminus D_\A(2^l t)} \exp(-\epsilon \|\vz\|)\\
    \le& \sum_{l=0}^\infty \sum_{\vz\in D_\A(2^{l+1}t)\setminus D_\A(2^lt)} \exp(-\epsilon 2^l t)\\
    \le& \sum_{l=0}^\infty |D_\A(2^{l+1}t)|\exp(-\epsilon 2^l t)\\
    \le& \sum_{l=0}^\infty L(2^{l+1}t)^{d-k}\exp(-\epsilon 2^l t)\tag{by \cref{eq:doubling2} and large enough $t$}\\
    =& L(2t)^{d-k}\sum_{l=0}^\infty 2^{l(d-k)}\exp(-\epsilon 2^l t)\\
    \le& 2L(2t)^{d-k}\exp(-\epsilon t)\tag{by ratio test and large enough $t$}
\end{align*}
Therefore, we complete the proof by taking $K = 2\cdot 2^{d-k} L =  \frac{4\cdot 2^{d-k}V_{d-k}}{\sqrt{\det(\mC_\A^\top\mC_\A)}}$.
\end{proof}

\section{Proofs for Gaussian mechanisms}\label{appendix:gaussian}
Similar to the Gaussian distribution on Euclidean space, a Gaussian distribution on lattices is sub-Gaussian (\cref{prop:gaussian}), with proof below.  This is non-trivial, because the support of Gaussian in \cref{eq:gen-gaussian} is not uniform.  For instance, suppose $\Lambda$ is not a lattice space, but has $2^{r^2}$ points in each integer radius $r$ shell.  Then the probability of Gaussian on such space is not sub-Gaussian.

\begin{proposition}[label = prop:gaussian]
Let $\mB\in \R^{d\times m}$ be a rank $m\le d$ matrix.  For any Gaussian $\rvz$ on $\Lambda = \Lambda(\mB)$ with $\vc = \vzero$ and $\sigma>0$, we have $\E[\rvz] = \vzero$ and 
$$\Pr[\|\rvz\|_2\ge t]\le Kt^m\exp\left(-\frac{t^2}{2\sigma^2}\right)$$
when $t>0$ is large enough and $K$ is defined in \cref{thm:lap_acc}.
\end{proposition}
\begin{proof}[Proof of \cref{prop:gaussian}]
The proof is mostly identical to \cref{thm:lap_acc}.  The unbiasedness is trivial.  Given $\Lambda \subseteq \R^d$, let $D(r) = \{\vz\in \Lambda: \|\vz\|\le r\}$.  The tail probability of $\vz$ is
$$\Pr[\|\vz\|\ge t] = \frac{\sum_{\vz\in \Lambda\setminus D(t)} \exp(-\frac{1}{2\sigma^2}\|\vz\|^2)}{\sum_{\vz\in \Lambda} \exp(-\frac{1}{2\sigma^2} \|\vz\|^2)}.$$
We will upper bound it by lower bounding the denominator and upper bounding the numerator. 

For the denominator, 
$$\sum_{\vz\in \Lambda} \exp(-\frac{1}{2\sigma^2}\|\vz\|^2)\ge \exp(-\epsilon \|\mathbf{0}\|^2) = 1.$$
For the numerator, by \cref{lem:doubling}, for all large enough $r>0$,
\begin{equation}\label{eq:doubling1}
    |D(r)|\le \frac{2V_{m}}{\sqrt{\det(\mB^\top\mB)}} r^{m}.
\end{equation}
Hence, if $t$ is large enough, we have
\begin{align*}
    &\sum_{\vz\in \Lambda\setminus D(t)} \exp(-\frac{1}{2\sigma^2}\|\vz\|^2)\\
    = &\sum_{l=0}^\infty \sum_{\vz\in D(2^{l+1}t)\setminus D(2^l t)} \exp(-\frac{1}{2\sigma^2}\|\vz\|^2)\\
    \le& \sum_{l=0}^\infty \sum_{\vz\in D(2^{l+1}t)\setminus D(2^lt)} \exp(-\frac{1}{2\sigma^2} 2^{2l} t^2)\\
    \le& \sum_{l=0}^\infty |D(2^{l+1}t)|\exp\left(-\frac{2^{2l} t^2}{2\sigma^2}\right)\\
    \le&  \frac{2V_{m}}{\sqrt{\det(\mB^\top\mB)}}\sum_{l=0}^\infty (2^{l+1}t)^{m}\exp\left(-\frac{2^{2l} t^2}{2\sigma^2}\right)\tag{by \cref{eq:doubling1} and large enough $t$}\\
    =& \frac{2(2t)^{m}V_{m}}{\sqrt{\det(\mB^\top\mB)}}\sum_{l=0}^\infty 2^{lm}\exp\left(-\frac{2^{2l} t^2}{2\sigma^2}\right)\\
    \le& \frac{2^{m+1}V_{m}}{\sqrt{\det(\mB^\top\mB)}}t^m\exp\left(-\frac{t^2}{2\sigma^2}\right)\tag{by ratio test and large enough $t$}
\end{align*}
Therefore, we complete the proof by taking $K = \frac{2^{m+1}V_{m}}{\sqrt{\det(\mB^\top\mB)}}$.
\end{proof}

\begin{proof}[Proof of \cref{thm:gaussian}]
By definition, the mechanism is $\A$-invariant. 
Now we show the mechanism is $(\epsilon,\delta)$ $\A$-induced integer subspace differentially private.  For all distinct $\vx\equiv_\A \vx'$ and a feasible outcome $\vy$, because the normalizing constants are identical, the privacy loss is 
\begin{align*}
    \ln \frac{\Pr[M(\vx) = \vy]}{\Pr[M(\vx') = \vy]} =& \ln \frac{\exp\left(-\frac{1}{2\sigma_{\epsilon, \delta}^2}\|\vy-\vx\|^2\right)}{\exp\left(-\frac{1}{2\sigma_{\epsilon, \delta}^2}\|\vy-\vx'\|^2\right)} \\
    &= \frac{1}{2\sigma_{\epsilon, \delta}^2}\left(\|\vy-\vx'\|^2-\|\vy-\vx\|^2\right),
\end{align*}
If we set $\vz:=\vy-\vx$, by cosine theorem the privacy loss is 
$$\ln \frac{\Pr[M(\vx) = \vy]}{\Pr[M(\vx') = \vy]}\le \frac{1}{2\sigma_{\epsilon, \delta}^2}\left(\|\vx-\vx'\|^2+2\|\vx-\vx'\|\|\vz\|\right).$$
Therefore, the privacy loss is bounded by $\epsilon \|\vx-\vx'\|$ when $\vz$ satisfies 
$$\|\vz\|\le\sigma_{\epsilon, \delta}^2\epsilon-\frac{1}{2}\|\vx-\vx'\|.$$
If the datasets $\vx$ and $\vx'$ are too far apart so that $\|\vx-\vx'\|\ge 2\sigma_{\epsilon, \delta}^2\epsilon = \frac{4c_\A\ln 1/\delta}{\epsilon}$, we cannot have any privacy guarantee as discussed in \cref{def:partitional_dp}.  On the other hand,  when $\|\vx-\vx'\|\le 2\sigma_{\epsilon, \delta}^2\epsilon$, let $r = \|\vx-\vx'\|$ and $t = \sigma_{\epsilon, \delta}^2\epsilon-\frac{1}{2}r$.  If we can show
\begin{equation}\label{eq:gaussian_privacy0}
    K t^{d-k}\exp\left(-\frac{t^2}{2\sigma_{\epsilon, \delta}^2}\right)\le e^{\epsilon(r-1)}\delta\le \frac{e^{ \epsilon r}-1}{e^\epsilon-1}\delta,
\end{equation}
we completes the proof by \cref{prop:gaussian}.

We first take log on both side of \cref{eq:gaussian_privacy0} and the above inequality is equivalent to
$\epsilon (r-1)-\ln 1/\delta+\frac{t^2}{2\sigma_{\epsilon, \delta}^2}-(d-k)\ln t-\ln K\ge 0$.

\begin{align}
    &\frac{t^2}{2\sigma_{\epsilon, \delta}^2}+\epsilon (r-1)-\ln 1/\delta-(d-k)\ln t-\ln K\nonumber\\
    =& \frac{1}{2} \sigma_{\epsilon, \delta}^2\epsilon^2-\frac{1}{2}\epsilon r+\frac{r^2}{8\sigma_{\epsilon, \delta}^2}+\epsilon (r-1)-\ln 1/\delta\\
    & -(d-k)\ln t-\ln K\tag{replace $t = \sigma_{\epsilon, \delta}^2\epsilon-\frac{1}{2}r$}\\
    \ge& \frac{1}{2} \sigma_{\epsilon, \delta}^2\epsilon^2-\epsilon-\ln 1/\delta-(d-k)\ln t-\ln K\nonumber\\
    =& c_\A \ln 1/\delta-\epsilon-\ln 1/\delta-(d-k)\ln t-\ln K\label{eq:gaussian_privacy1}
\end{align}
Now we show \cref{eq:gaussian_privacy1} is greater than zero for some $c_\A = \mathcal{O}(\max\{(d-k)\ln (d-k), \ln K\})$ that prove \cref{eq:gaussian_privacy0}.  We show this in three parts.  First because $\delta<\epsilon<e^{-\epsilon}$, $\ln 1/\delta-\epsilon>0$.  Second, for some $c_\A = \mathcal{O}(\max\{(d-k)\ln (d-k), \ln K\})$, we have $\frac{1}{3}c_\A\ln 1/\delta\ge \frac{c_\A}{3}>\ln K$.  Finally, 
\begin{align*}
    &\left(\frac{2}{3}c_\A-1\right)\ln 1/\delta-(d-k)\ln k\\
    \ge& \left(\frac{2}{3}c_\A-1\right)\ln 1/\delta-(d-k)\ln \sigma_{\epsilon, \delta}^2\epsilon\tag{because $t \le \sigma_{\epsilon, \delta}^2\epsilon$}\\
    =& \left(\frac{2}{3}c_\A-1\right)\ln 1/\delta-(d-k)\left(\ln\ln 1/\delta+\ln 1/\epsilon+\ln c_\A\right)\\
      \ge& \left(\frac{2}{3}c_\A-1-2(d-k)\right)\ln 1/\delta-(d-k)\ln c_\A\tag{$\delta<\epsilon<1/e$}\\
      \ge& 0
\end{align*}

The last inequality holds for some $c_\A = O((d-k)\ln (d-k))$ that completes the proof.
\end{proof}

\section{Details on Implementation}\label{appendix:mcmc}

\subsection{The Gibbs-within-Metropolis sampler}

~\cref{alg:metropolis-gibbs} presents a Gibbs-within-Metropolis sampler, which produces a sequences of dependent draws $\left(\vz^{\left(l\right)}\right)_{0\le l\le \text{nsim}}$ from the target distribution $q_{\epsilon}$, using an additive jumping distribution whose element-wise construction is as described in~\cref{eq:proposal}.  

\begin{algorithm}[htb]
\caption{A Gibbs-within-Metropolis sampler}
\label{alg:metropolis-gibbs}
\begin{algorithmic}[1] 
\STATE \textbf{input}:  initial distribution $\pi_{0}$, pre-jump distribution $\{\eta_{j}\}_{k+1\le j\le n}$, target distribution $q_{\epsilon}$, $\mV  \in \Z^{d\times d}$ from the Smith normal form of constraint $\mA \in \Z^{k\times d}$ of rank $k$;
\STATE Initialize $\vv^{\left(0\right)} \sim \pi_{0}$; 
\FOR{$l=1,\ldots,\text{nsim}$}
\STATE set $\vv^{*} = \vv^{\left(l - 1\right)}$;
\FOR{$j=k+1,\ldots,n$}
\STATE sample ${e}_{j}^{*}\sim \eta_{j}$ and update $\vv^{*}_{[j]}=\vv^{\left(l-1\right)}_{[j]} +e^{*}_{j}$;
\ENDFOR
\STATE sample $r \sim U[0, 1]$;
\STATE if $r \le {q_{\epsilon}\left(\mV\vv^{*}\right)}/{q_{\epsilon}\left(\mV\vv^{\left(l-1\right)}\right)}$ then set  $\vv^{\left(l\right)}=\vv^{*}$,
otherwise set $\vv^{\left(l\right)}=\vv^{\left(l-1\right)}$;
\STATE set $\vz^{\left(l\right)} = \mV\vv^{\left(l\right)}$;
\ENDFOR
\STATE \textbf{return}: a chain $\left(\vz^{\left(l\right)}\right)_{0\le l\le \text{nsim}}$ following the single transition kernel $\vz^{\left(l\right)} \sim K\left(\vz^{\left(l-1\right)},\cdot \right)$  with target distribution $q_\epsilon$.
\end{algorithmic}
\end{algorithm}

The choice of $\eta_j$ is a tuning decision for the algorithm. If using the double geometric proposal as described in ~\cref{sec:implementation}, the scale parameter $a$ should be set to encourage fast mixing of the chain. For the examples in ~\cref{sec:experiments}, we report the specific choices as they occur, as well as in ~\cref{appendix:experiments}. The initial distribution $\pi_0$ may be chosen simply as $g$ for convenience. 

As noted in ~\cref{sec:implementation}, steps 5 through 7 of ~\cref{alg:metropolis-gibbs} updates the proposed jump in a Gibbs sweep (i.e. one dimension at a time), utilizing the fact that the pre-jump object $\ve$ is element-wise independent under $g_\mA$. Doing so facilitates the $L$-lag coupling for assessing empirical convergence of the chain. In practice, these updates may also be performed simultaneously.



\subsection{$L$-lag maximal coupling algorithm and meeting time sampler.}

We adapt the method proposed by \citep{biswas2019estimating} to perform empirical convergence assessment of the proposed Metropolis algorithm using $L$-lag coupling.
As ~\cref{sec:implementation} explains, we construct a joint transition kernel $\tilde{K}$ of two Markov chains, each having the same target distribution $q_\epsilon$ induced by the marginal transition kernel $K$ as defined in ~\cref{alg:metropolis-gibbs}. The joint kernel $\tilde{K}$ is a maximal coupling kernel and is given in ~\cref{alg:coupling-gibbs}. It is designed in such a way that the $L$-th lag of the two chains will couple in finite time with probability one. This $L$-lag meeting time is a random variable, and is sampled via ~\cref{alg:llag-alt}, which employs ~\cref{alg:metropolis-gibbs} and~\ref{alg:coupling-gibbs} as subroutines.


\begin{algorithm}[htb]
\caption{A maximal $L$-lag Gibbs-within-Metropolis coupling
}
\label{alg:coupling-gibbs}
\begin{algorithmic}[1] 
\STATE \textbf{input}: lag $L \ge 1$, iteration $l > L$,  states $\left(\vz^{\left(l-1\right)},\vz_{\circ}^{\left(l-L-1\right)}\right)$
\STATE Initialize $\vv^{\left(l-1\right)} = \mV^{-1}\vz^{\left(l-1\right)}$, $\vw^{\left(l-L-1\right)} = \mV^{-1}\vz_{\circ}^{\left(l-L-1\right)}$
\FOR{$j=k+1,\ldots,n$}
\STATE set $\vv^* = \vv^{\left(l-1\right)}$ and $\vw^* = \vw^{\left(l-L-1\right)}$;
\STATE sample ${e}_{j}^{*}\sim \eta_{j}$ and update $\vv^{*}_{[j]}=\vv^{\left(l-1\right)}_{[j]} +e^{*}_{j}$;
\STATE sample $s \sim U[0, 1]$; 
\STATE \textbf{if} $s \eta_{j}\left( e^{*}_j \right) \le \eta_j\left( \vv^{*}_{[j]} - \vw^{(l-L-1)}_{[j]}\right)$ \textbf{then}  update $\vw^{*}_{[j]} = \vv^{*}_{[j]}$; \\
\textbf{else} sample $\tilde{e}_j \sim \eta_{j}$, $\tilde{s}\sim U[0,1]$, and set $\tilde{w}_j = \vw^{(l-L-1)}_{[j]} +  \tilde{e}_j$  until $\tilde{s} \eta_j\left(\tilde{e}_j\right) > \eta_j\left( \tilde{w}_j  - \vv^{(l-1)}_{[j]}\right)$; update $\vw^{*}_{[j]} = \tilde{w}_j$. 
\ENDFOR
\STATE sample $r \sim U[0, 1]$;
\STATE \textbf{if}  $r \le {q_{\epsilon}\left(\mV\vv^{*}\right)}/{q_{\epsilon}\left(\mV\vv^{\left(l-1\right)}\right)}$ \textbf{then} set  $\vv^{\left(l\right)}=\vv^{*}$, \textbf{else} set  $\vv^{\left(l\right)}=\vv^{\left(l-1\right)}$; 
\STATE \textbf{if}  $r \le {q_{\epsilon}\left(\mV\vw^{*}\right)}/{q_{\epsilon}\left(\mV\vw^{\left(l-L-1\right)}\right)}$ \textbf{then} set  $\vw^{\left(l\right)}=\vw^{*}$, \textbf{else} set  $\vw^{\left(l-L\right)}=\vw^{\left(l-L-1\right)}$;
\STATE set $\vz^{\left(l\right)} = \mV\vv^{\left(l\right)}$, $\vz_{\circ}^{\left(l-L\right)} = \mV\vw^{\left(l-L\right)'}$.
\STATE \textbf{return}: a 
pair of draws $\left(\vz^{\left(l\right)}, \vz_{\circ}^{\left(l-L\right)}\right)$ following the joint transition kernel  $\tilde{K}\left(\left(\vz^{\left(l-1\right)},\vz_{\circ}^{\left(l-L-1\right)}\right), \cdot\right)$, each with marginal target distribution $q_\epsilon$.
\end{algorithmic}
\end{algorithm}


The output of ~\cref{alg:llag-alt} is used to construct an estimate of the upper bound on the total variation distance between the target distribution  $q_{\epsilon}$ and the marginal distribution of the chain at time $l$, denoted $q_{\epsilon}^{(l)}$ \citep[Theorem 2.5]{biswas2019estimating}. The upper bounds are obtained as empirical averages over independent runs of coupled Markov chains. 

\begin{algorithm}[htb]
\caption{$L$-lag meeting time sampler (adapted from \citep{biswas2019estimating})}
\label{alg:llag-alt}
\begin{algorithmic}[1] 
\STATE \textbf{input}: lag $L \ge 1$, initial distribution $\pi_{0}$, single kernel $K$ and joint kernel $\tilde{K}$
\STATE \textbf{output}: meeting time $\tau^{(L)}$, chains $\left(\vz^{(l)}\right)_{0\le l\le \tau^{(L)}}$, $\left(\vz_{\circ}^{(l)}\right)_{0\le l\le \tau^{(L)}-L}$
\STATE Initialize $\vz^{(0)} \sim \pi_{0}$,   $\vz^{\left(l+1\right)} \mid \vz^{\left(l\right)} \sim K\left(\vz^{\left(l\right)}, \cdot\right)$ for $1 \le l \le L$, and $\vz_{\circ}^{(0)} \sim \pi_{0}$
\FOR{$l > L$}
\STATE sample 
$\left(\vz^{\left(l\right)}, \vz_{\circ}^{\left(l-L\right)}\right) \sim \tilde{K}\left(\left(\vz^{\left(l-1\right)},\vz_{\circ}^{\left(l-L-1\right)}\right), \cdot\right)$
\STATE \textbf{if} $\vz^{\left(l\right)} = \vz_{\circ}^{\left(l-L\right)}$ \textbf{then return} $\tau^{(L)} := l$  and  chains $\left(\vz^{(l)}\right)_{0\le l\le \tau^{(L)}}$ and $\left(\vz_{\circ}^{(l)}\right)_{0\le l\le \tau^{(L)}-L}$
\ENDFOR
\end{algorithmic}
\end{algorithm}

For ~\cref{alg:coupling-gibbs}, one needs to choose a lag $L>0$. As \citep{biswas2019estimating} discusses, smaller $L$ encourages faster coupling of the two chains, however the estimated total variation distance as a upper bound is a poorer one. On the other hand, larger $L$ produces a a tighter bound at the expense of a heavier computational burden. \citep{biswas2019estimating} recommends starting with $L = 1$, and increasing it to a point where the estimated TV upper bound $d_{TV}(q_\epsilon^{(0)}, q_{\epsilon})$ is close to 1 (i.e. non-vacuous). \citep{biswas2019estimating} also advises against increasing $L$ arbitrarily, because the TV upper bound itself is not an optimal one thus the benefit of sharpness is limited. In our experiments in ~\cref{sec:experiments}, we employ $L$ values such that the distribution of the $L$-lag coupling time, $\tau^{(L)} - L$, appears to be empirically stable in $L$. The choices of $L$ are reported as they occur, as well as in ~\cref{appendix:experiments}.

\section{Additional experimental results and details}\label{appendix:experiments}



\subsection{Additional details}

The code and data for reproducing plots is included in the supplementary material. All experiments were run on a Macbook with M1 processor and 16GB memory. The packages and software used are all available freely.

\subsection{Intersecting counting constraints.} 



Consider a synthetic example in which a set of three counting constraints $\mathcal{A} = (A_1, A_2, A_3)$ are defined over 14 records, where the intersection of all subsets of constrains are nonempty. The constraints are schematically depicted by ~\cref{fig:circles-constraint}. 
The corresponding incident matrix is $\mA \in \Z^{3\times 14}$, with each row corresponding to one of $A_1, A_2$ and $A_3$, and each element being $1$ at indices corresponding to the record within that constraint and $0$ otherwise.  We apply ~\cref{alg:metropolis-gibbs} to instantiate the generalized Laplace mechanism with both $\ell_1$-norm and $\ell_2$-norm targets, to privatize a data product that conforms to the constraint $\mA$ with  $\epsilon = 0.25$. The pre-jump proposal distributions $\eta_j$ are double geometric distributions, with parameter $a=\exp{(-1)}$ for the $\ell_1$-norm target and $a=\exp{(-1.5)}$ for the $\ell_2$-norm target. ~\cref{fig:circles-coupling} shows the estimated upper bound total variation distance using $L$-lag coupling. The chains appear to empirically converge at around $10^5$ iterations for the $\ell_1$-norm target distribution and $10^6$ iterations for the $\ell_2$-norm target distribution.

\begin{figure}[t]
    \centering
      \includegraphics[trim={60mm 0 60mm 33mm},clip,width=.5\textwidth]{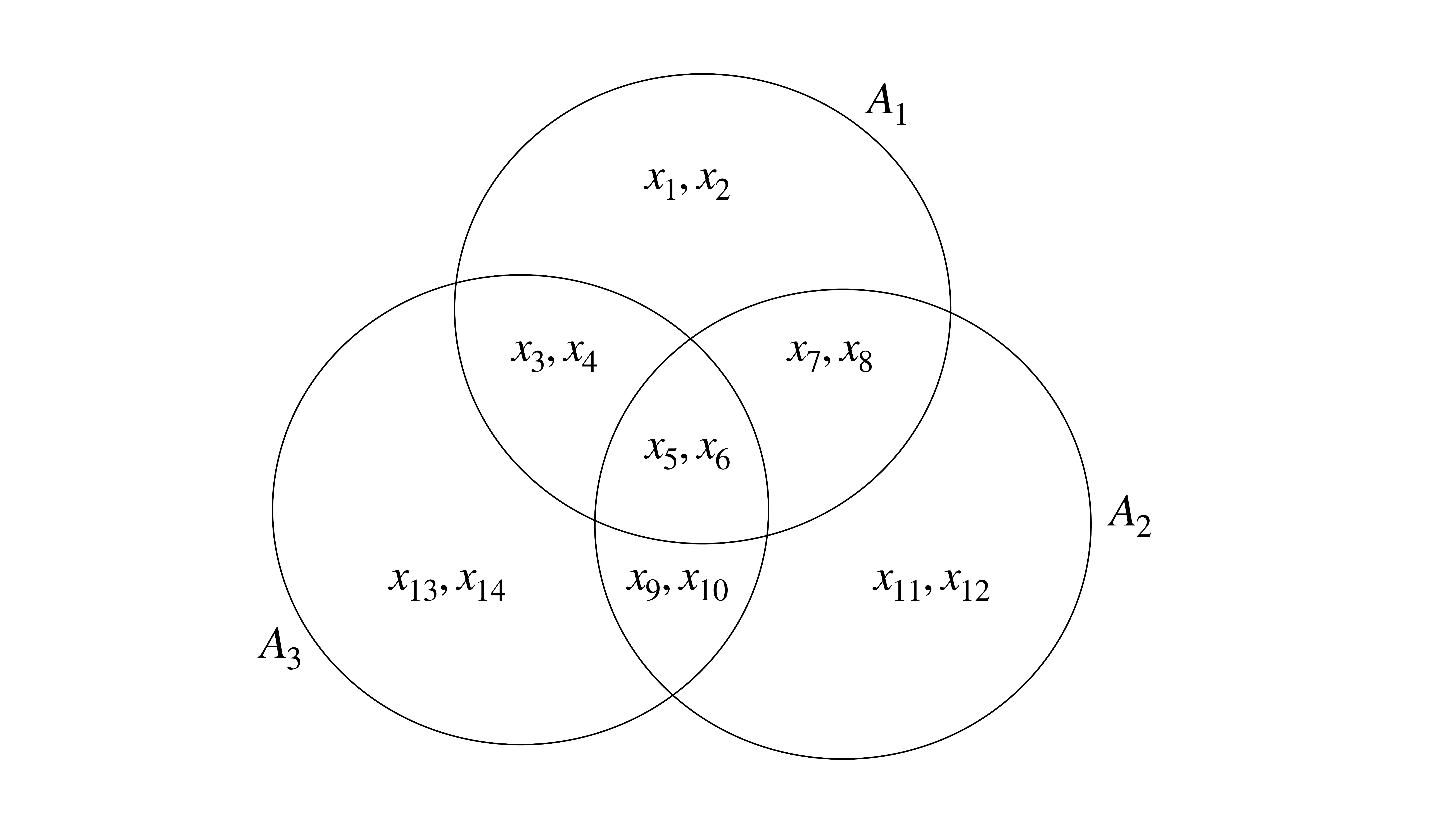}
  \caption{Schematic depiction of the intersecting counting constraints $\mathcal{A} = (A_1, A_2, A_3)$.}
    \label{fig:circles-constraint}
\end{figure}

\begin{figure}[h]
  \centering
  \begin{minipage}[b]{0.49\textwidth}
  \centering
    \includegraphics[width=0.95\textwidth]{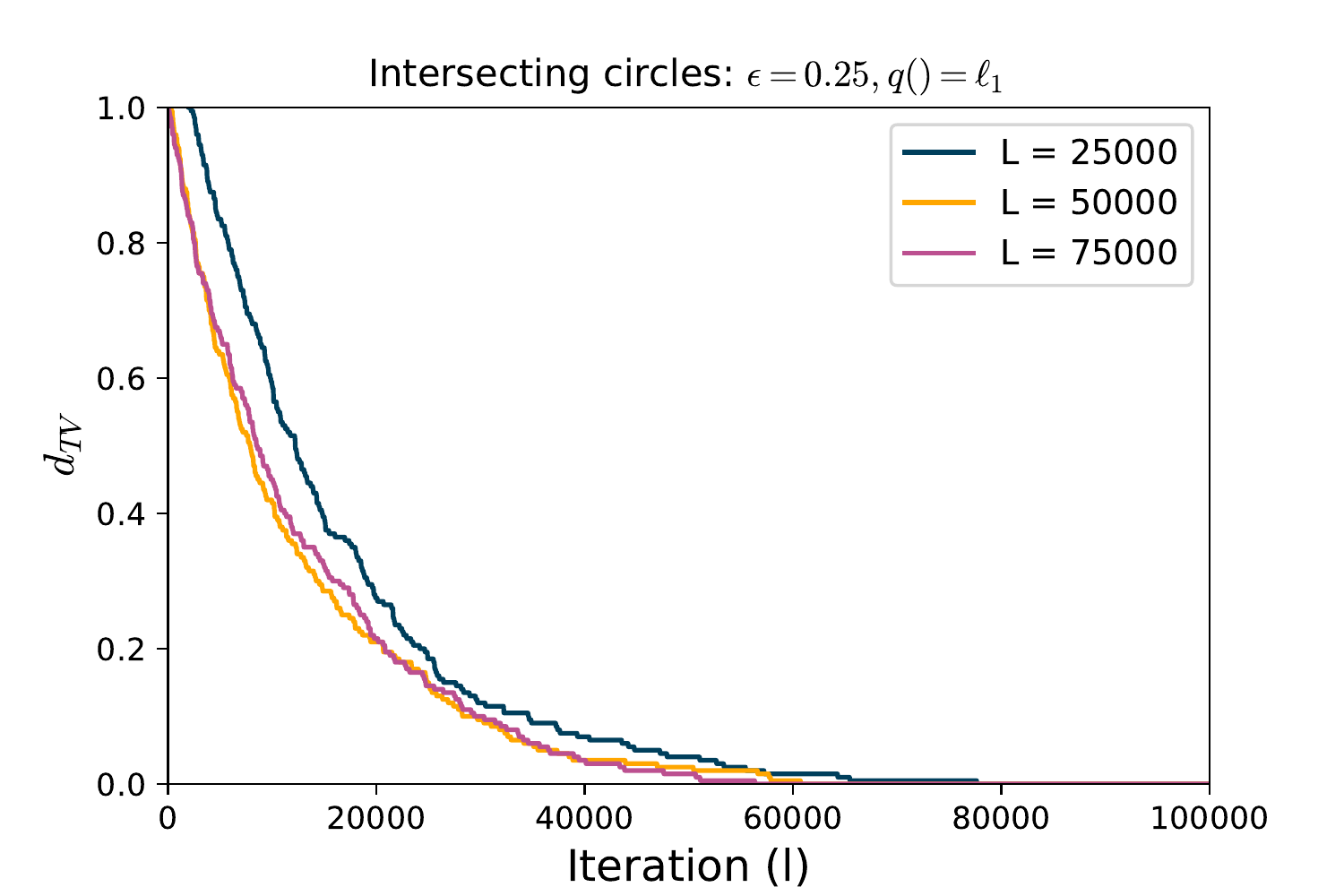}
  \end{minipage}
  \begin{minipage}[b]{0.49\textwidth}
  \centering
    \includegraphics[width=0.95\textwidth]{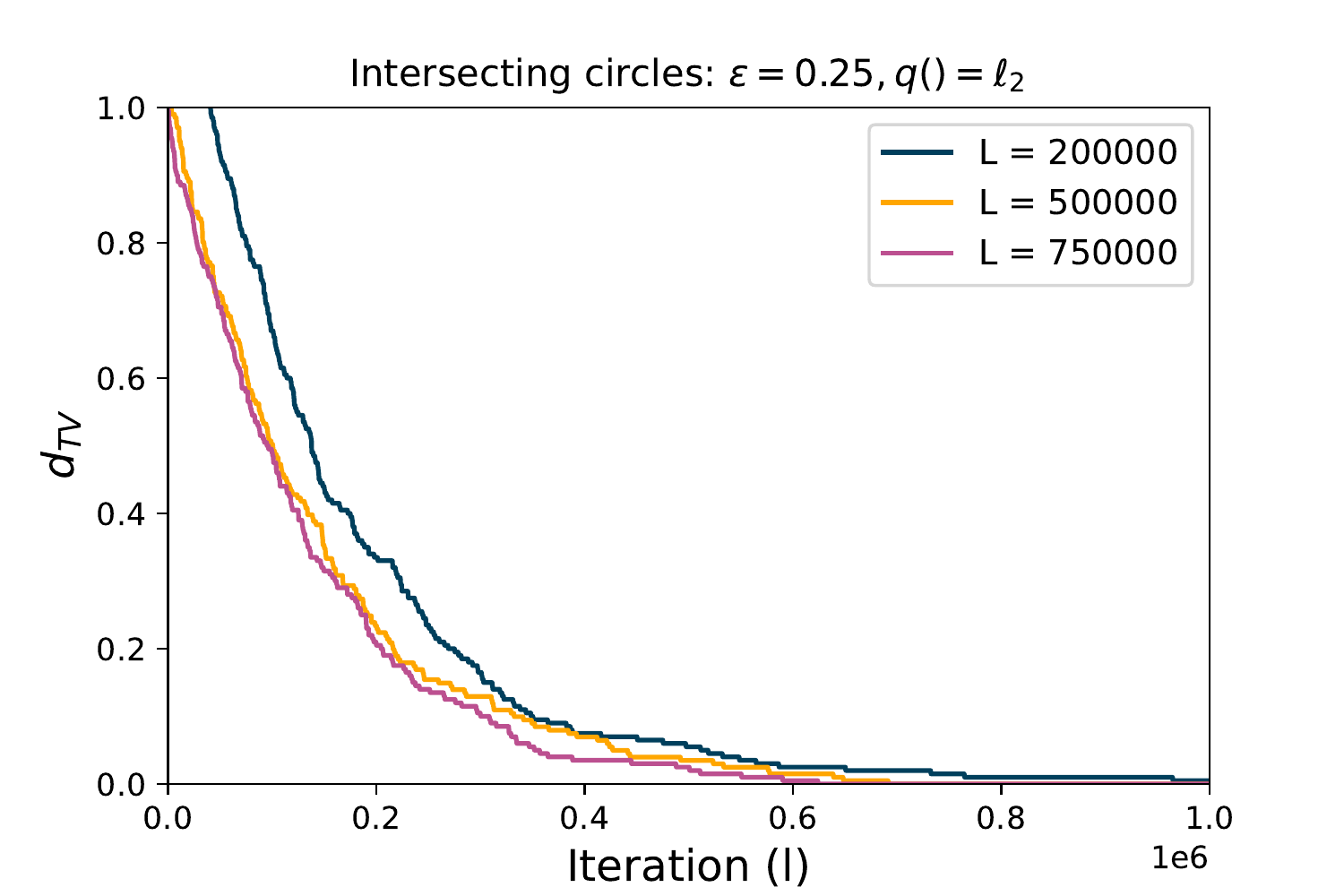}
  \end{minipage}
  \caption{Intersecting constraints: estimated upper bound on the chain's total variation distance from target distribution using $L$-lag coupling \citep{biswas2019estimating}. Left: $\ell_1$-norm target; right: $\ell_2$-norm target by the generalized Laplace mechanism with  $\epsilon = 0.25$.}
  \label{fig:circles-coupling}
\end{figure}

~\cref{fig:circle-boxplot} shows the boxplot of 1000  realizations of the proposed entry-wise errors respectively under the $\ell_1$-norm and $\ell_2$-norm target distributions.  The proposed additive noise are integer-valued and are entry-wise unbiased. 

\begin{figure*}[h]
    \centering
    \includegraphics[width=0.75\textwidth]{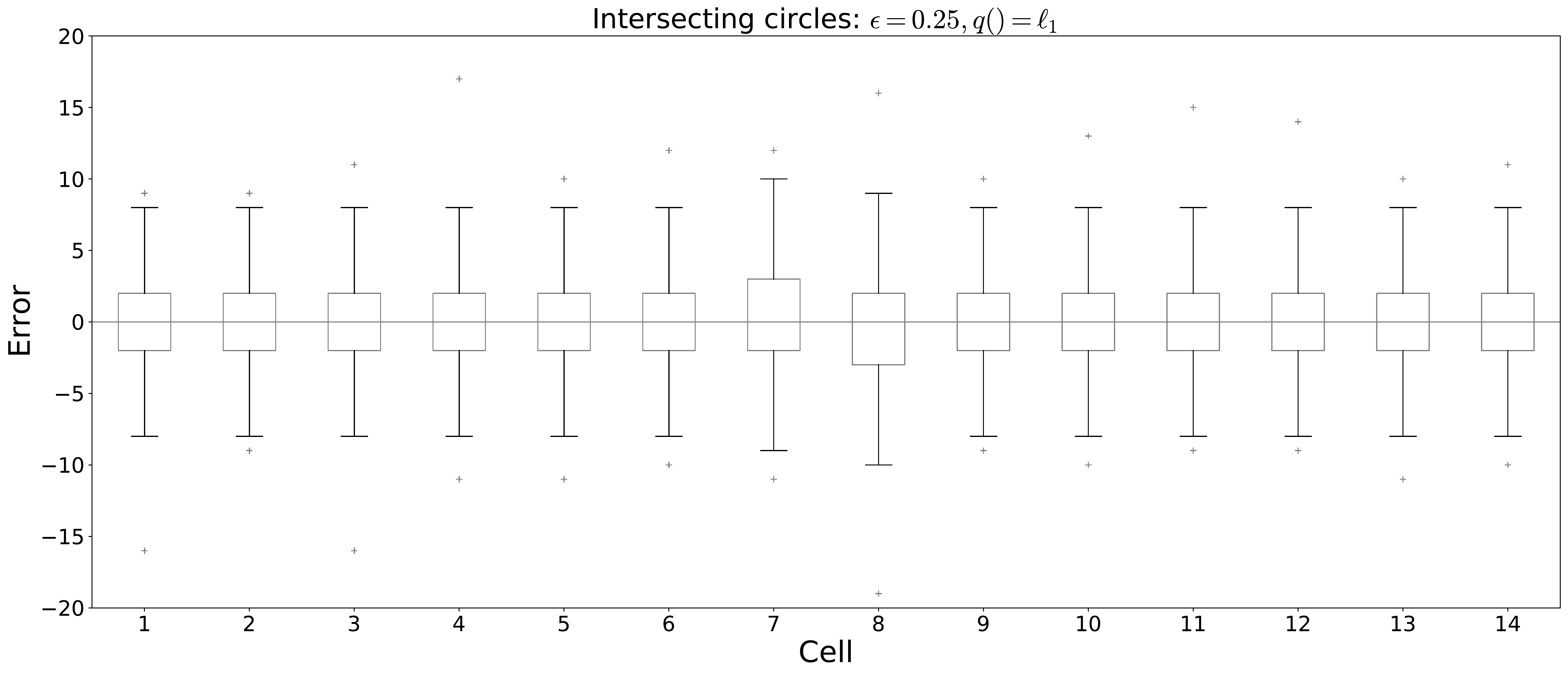}
    \includegraphics[width=0.75\textwidth]{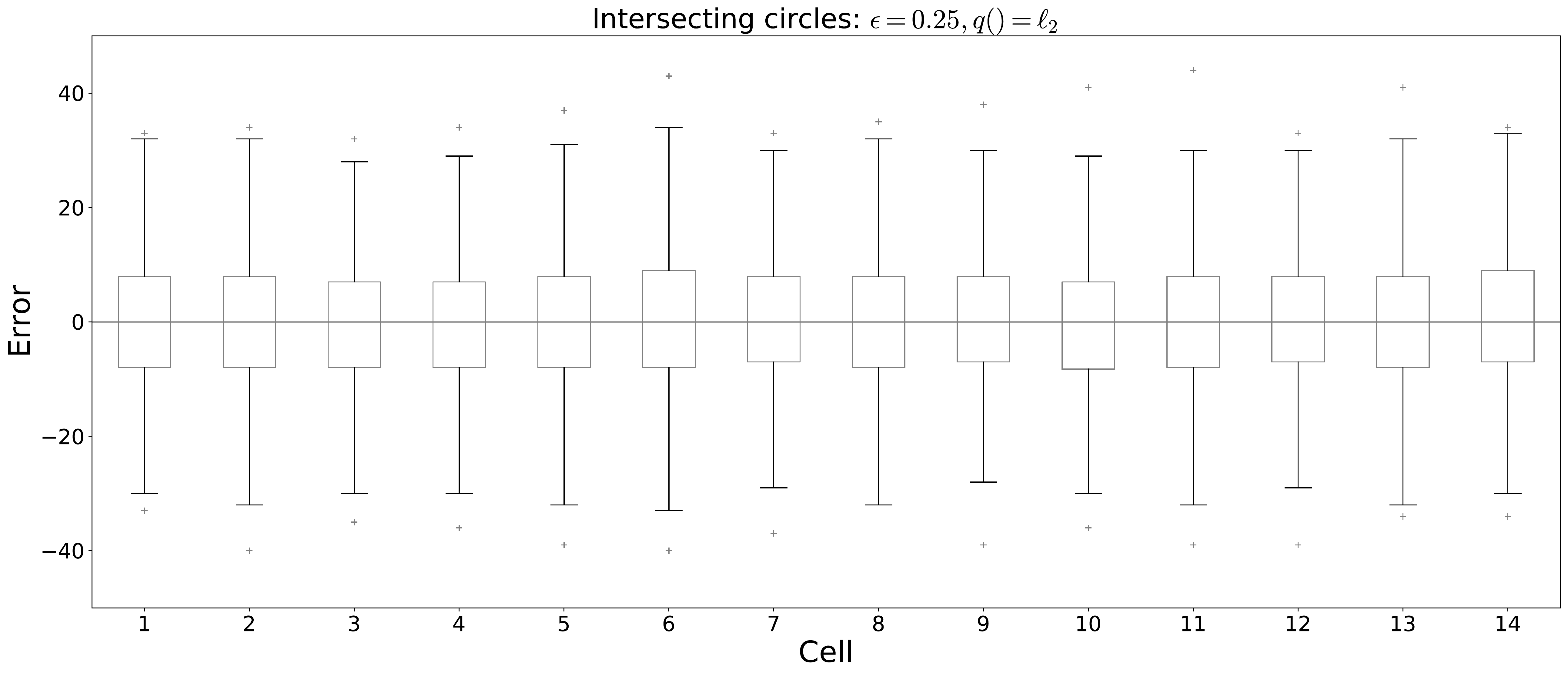}
    \caption{Intersecting counting constraints: boxplots of 1000 realizations of proposed entrywise errors; $\ell_1$-norm target distribution (top) and $\ell_2$-norm target distribution (bottom), $\epsilon=0.25$.}
    \label{fig:circle-boxplot}
\end{figure*}

\subsection{Delinquent Children by County and Household Head Education Level.}

In the \emph{report on statistical disclosure limitation (SDL) methodology}, the Federal Committee on Statistical Methodology published a fictitious dataset concerning delinquent children in the form of a $4\times 4$ contingency table, tabulated across four counties by education level of household head \cite[Table 4]{fcsm2005sdl}. Reproduced in ~\cref{table:children}, the dataset was used to illustrate various traditional SDL techniques, and was employed again by \citep{slavkovic2010synthetic} in a comparative study on the effect of swapping versus synthetic table generation on disclosure risk and on the validity of downstream statistical analysis.

\begin{table}[h]
\centering
\caption{$4\times 4$ contingency table: delinquent children by county and household head education level (reproduced from Table 4 of \citep{fcsm2005sdl})}
\setlength\tabcolsep{3pt}
\begin{tabular}{cccccc}
     \toprule
     County & Low & Medium & High & Very High & Total \\
     \midrule
     Alpha & 15 & 1 & 3 & 1 & 20 \\
     Beta & 20 & 10 & 10 & 15 & 55 \\
     Gamma & 3 & 10 & 10 & 2 & 35 \\
     Delta & 12 & 14 & 7 & 2 & 35 \\
     \midrule
     Total & 50 & 35 & 30 & 20 & 135 \\
     \bottomrule
\end{tabular}
\label{table:children}
\end{table}

Consider the privatization of the data while preserving the margins of the contingency table. We apply ~\cref{alg:metropolis-gibbs} to instantiate the generalized Laplace mechanism with both $\ell_1$-  and $\ell_2$-norm targets and with $\epsilon = 0.25$. The pre-jump proposal distributions $\eta_j$ are double geometric distributions, with parameter $a = \exp(-1)$ for the  $\ell_1$-norm target distribution, and $a = \exp(-2)$ for the $\ell_2$-norm target distribution.

~\cref{fig:children-coupling-l1} in ~\cref{sec:experiments}  show the evolution of the TV upper bounds on the chain's marginal distributions to the $\ell_1$-norm target distribution, each estimated with 200 independent coupled chains, and ~\cref{table:children-noise} shows one realization of the noises obtained after the chain has achieved empirical convergence under the $\ell_1$ norm target. ~\cref{table:children-noise-l2} and ~\cref{fig:children-coupling-l2} presented here are analogs of ~\cref{table:children-noise} and ~\cref{fig:children-coupling-l1}, respectively, but for the $\ell_2$-norm target distribution. Under  the $\ell_2$-norm target, the chains appear to converge after about $10^5$ iterations, and are stable at various choices of $L$ as shown in the figures.

\begin{center}

\begin{tabular}{cccccc}
     \toprule
     County & Low & Medium & High & Very High & Total \\
     \midrule
     Alpha & -11 & 1 & 10 & 0 & 0 \\
     Beta & -3 & -6 & 1 & 8 & 0 \\
     Gamma & -6 & -5 & 13 & -2 & 0 \\
     Delta & 20 & 10 & -24 & -6 & 0 \\
     \midrule
     Total & 0 & 0 & 0 & 0 & 0 \\
     \bottomrule
\end{tabular}
\captionof{table}{
Proposed generalized Laplace additive noise  ($\epsilon=0.25$, $\ell_2$-norm target distribution) ~\cref{alg:metropolis-gibbs}, which preserves row and column margins.\label{table:children-noise-l2}}
\end{center}
\begin{figure}[h]
    \centering
   \includegraphics[width=0.5\textwidth]{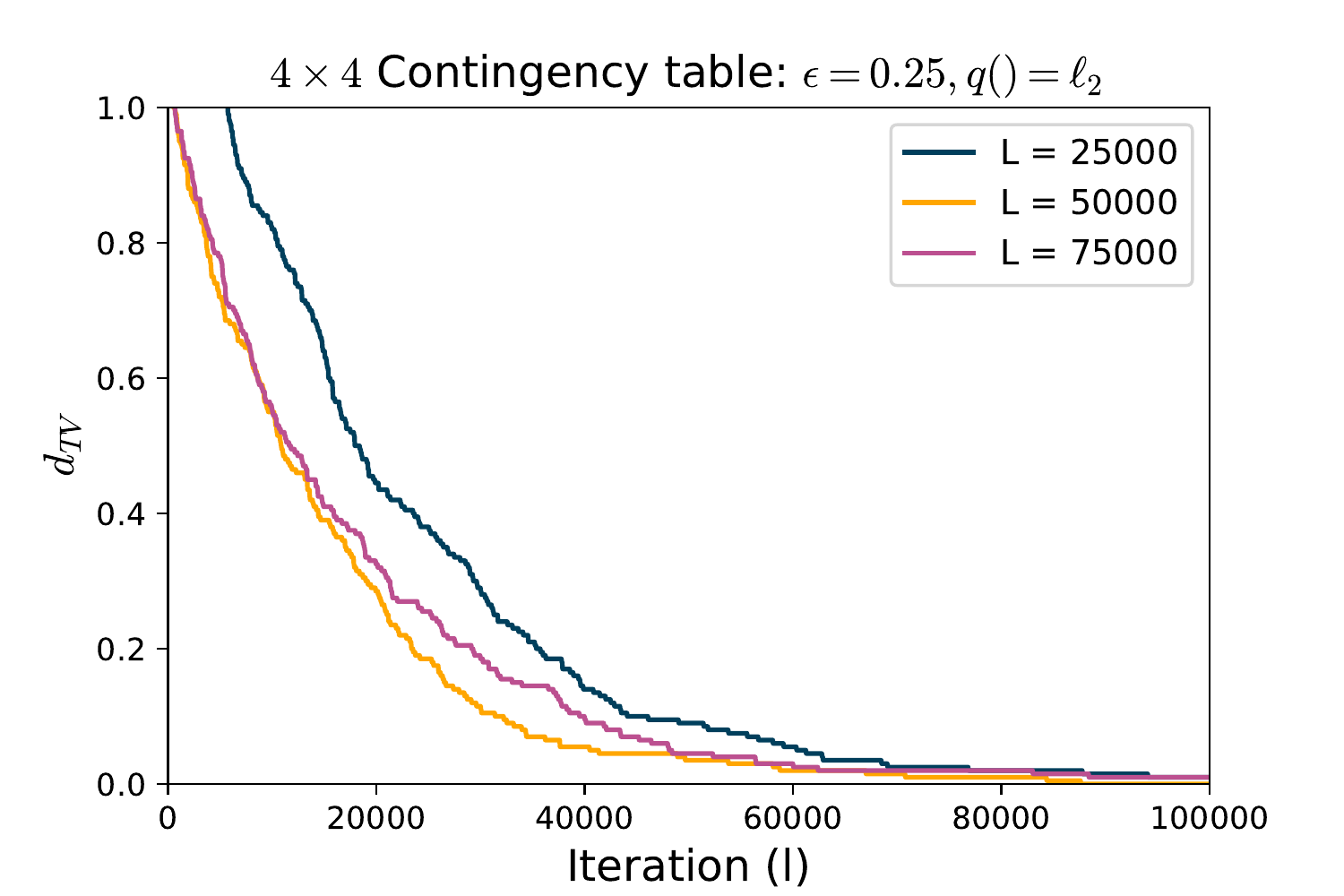}
  \captionof{figure}{Estimated upper bound: TV distance from target ($\ell_2$ norm) 
    \label{fig:children-coupling-l2}}
\end{figure}



In addition, ~\cref{fig:chidren-boxplot} shows the boxplot of 1000 realizations (thinned at 0.01\%) of the proposed errors under both the $\ell_1$-norm and $\ell_1$-norm target distributions. The proposed additive noise are integer-valued and are cellwise unbiased.

\begin{figure*}[b]
    \centering
    \includegraphics[width=0.75\textwidth]{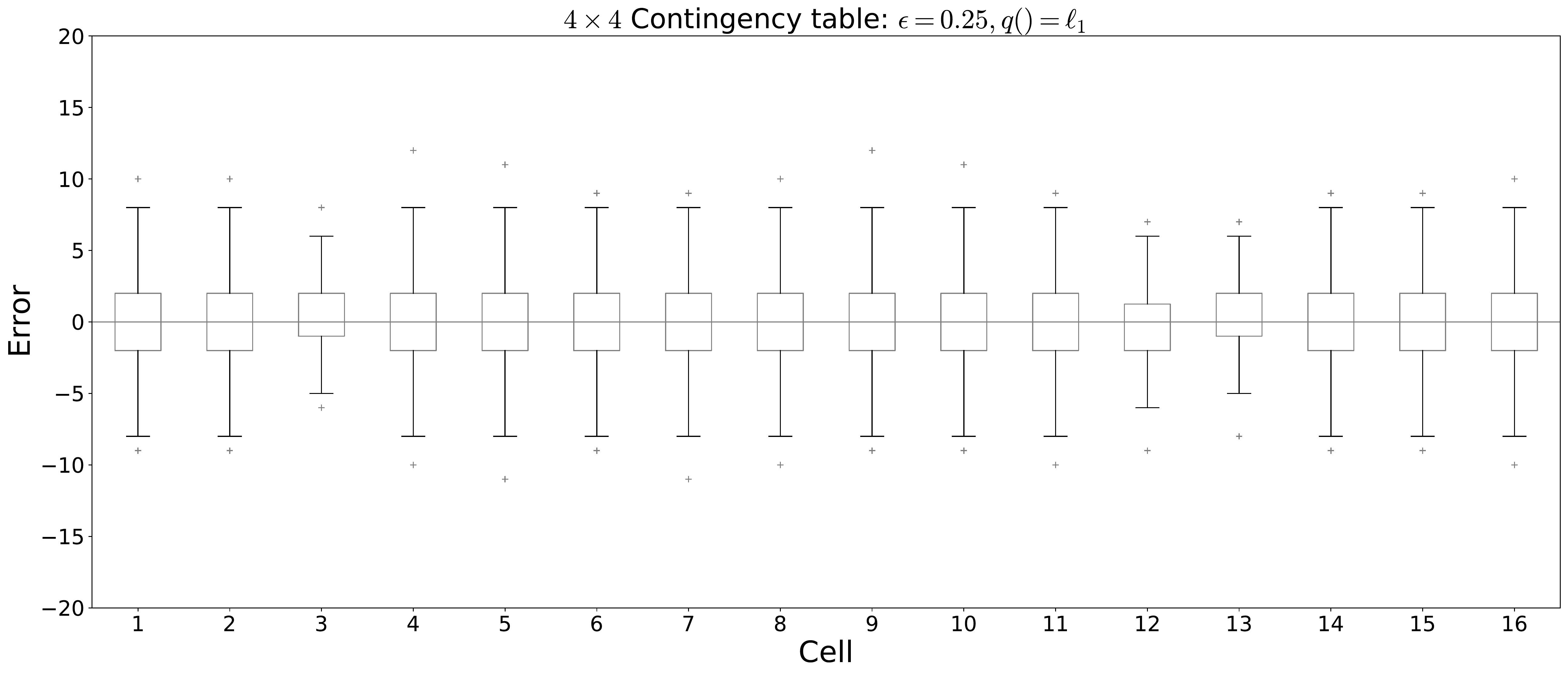}
    \includegraphics[width=0.75\textwidth]{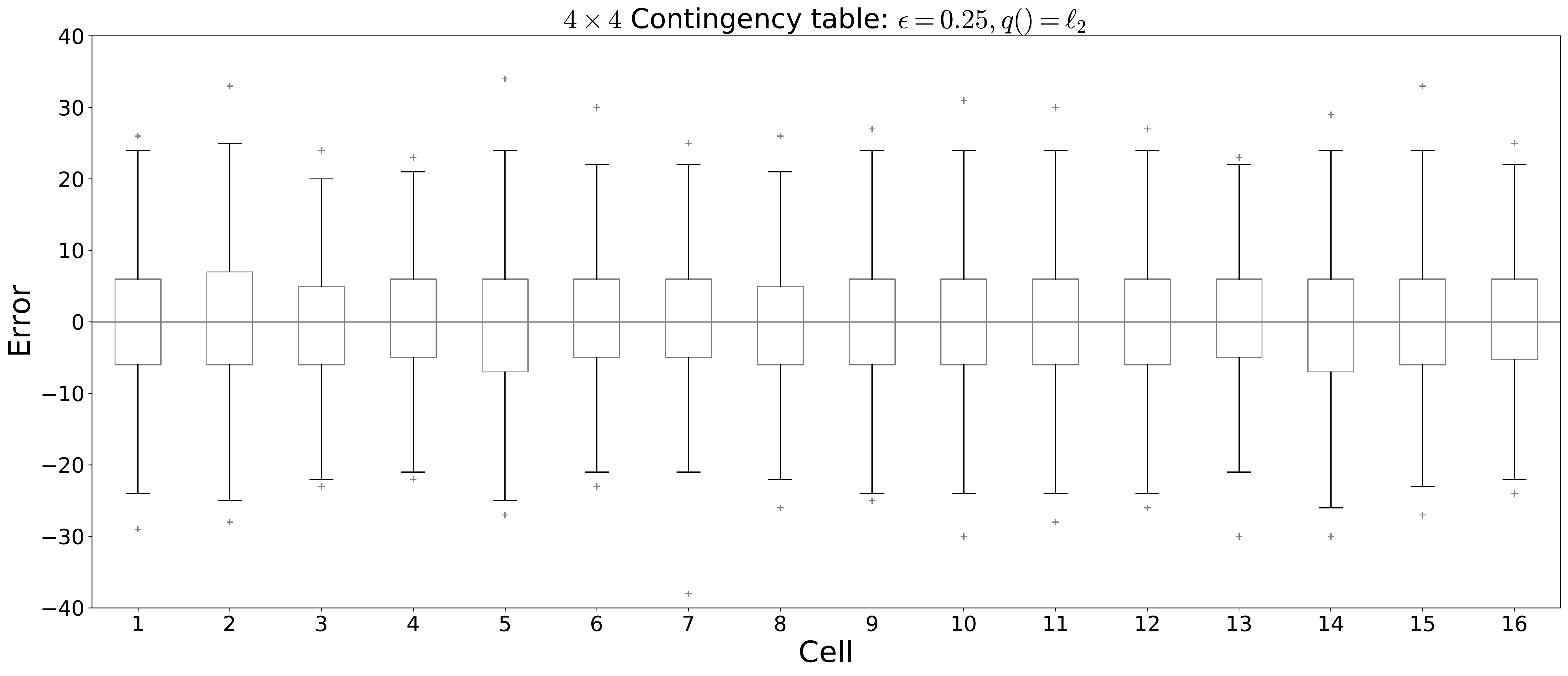}
    \caption{$4\times 4$ contingency table: boxplot of 1000 realizations of proposed cell-wise errors; $\ell_1$-norm target distribution (top) and $\ell_2$-norm target distribution (bottom), $\epsilon=0.25$.}
    \label{fig:chidren-boxplot}
\end{figure*}

\subsection{ 2010 U.S. Census county-level population data.}

We consider the publication of county-level population counts subject to the invariant of state population size. The U.S. Census Bureau published iterations of privacy-protected demonstration files produced by preliminary versions of its 2020 Disclosure Avoidance System (DAS), which treats the 2010 Census Summary Files (CSF) as the confidential values. These data have been curated by IPUMS NHGIS  and are publicly available \citep{ipums2020das}. Employed here are the November 2020 vintage demonstration data, protected by pure differential privacy with $\epsilon = 0.192 = 4 \text{ (total)} \times 0.16 \text{ (county level)} \times 0.3 \text{ (population query)}$.

We demonstrate ~\cref{alg:coupling-gibbs} using the generalized Laplace mechanism under $\ell_1$ norm, with $\epsilon$ set to accord to the Census Bureau's specification. The states that are used for demonstration in this work are identified by~\citep{gao2022subspace} as those for which the Census DAS produced statistically significantly biased errors  (at the $\alpha = 0.01$ level) for the county populations. There are 11 states in total. The result for the state of Illinois is presented in ~\cref{fig:illinois} in ~\cref{sec:intro}. ~\cref{fig:census_sdp_county_combined} below presents results for 10 additional states.

To assess the empirical convergence of the sampler, for each state we initialize four independent chains from over-dispersed starting values. To obtain these starting values, independent chains are run for $10^{6}$ iterations with an over-dispersed target distribution, specifically $q_{\epsilon}$  with $\epsilon=0.1$. The pre-jump proposal parameter is set to $a = \exp(-2.5)$ and is the same for all the states. With those starting values, each independent chains again expends $10^6$ steps of burn-in with the desired target distribution, $q_{\epsilon}$ with $\epsilon=0.192$. The potential scale reduction factors calculated across independent chains  are observed to reach $< 1.01$ for each county of each state.








\begin{figure*}
\begin{center}
\includegraphics[width=.85\textwidth]{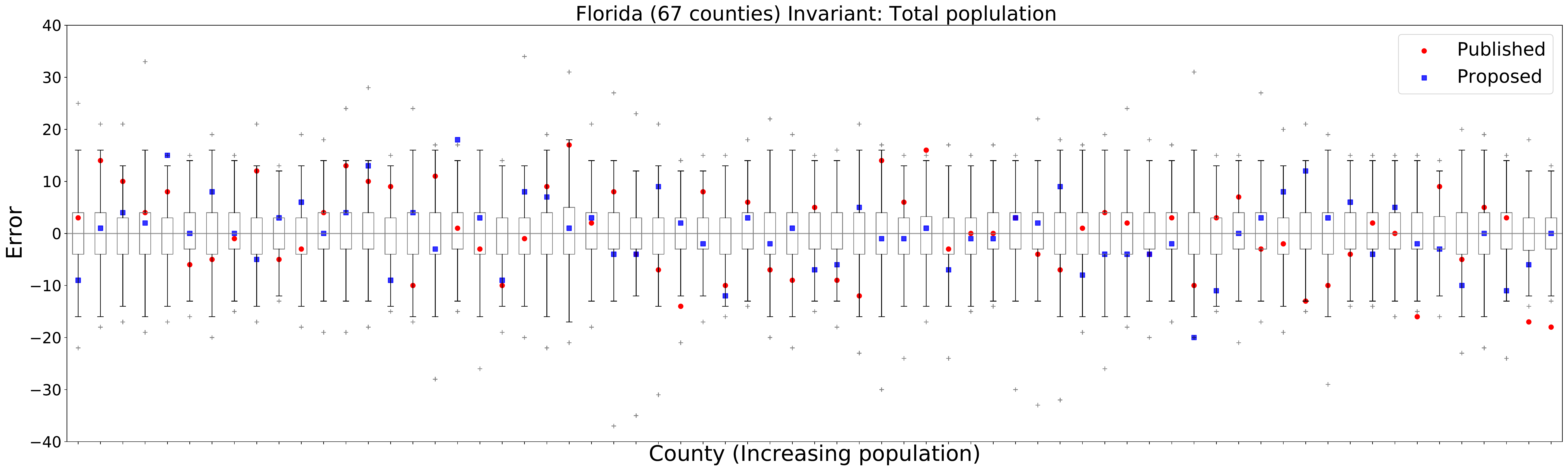}
\includegraphics[width=.85\textwidth]{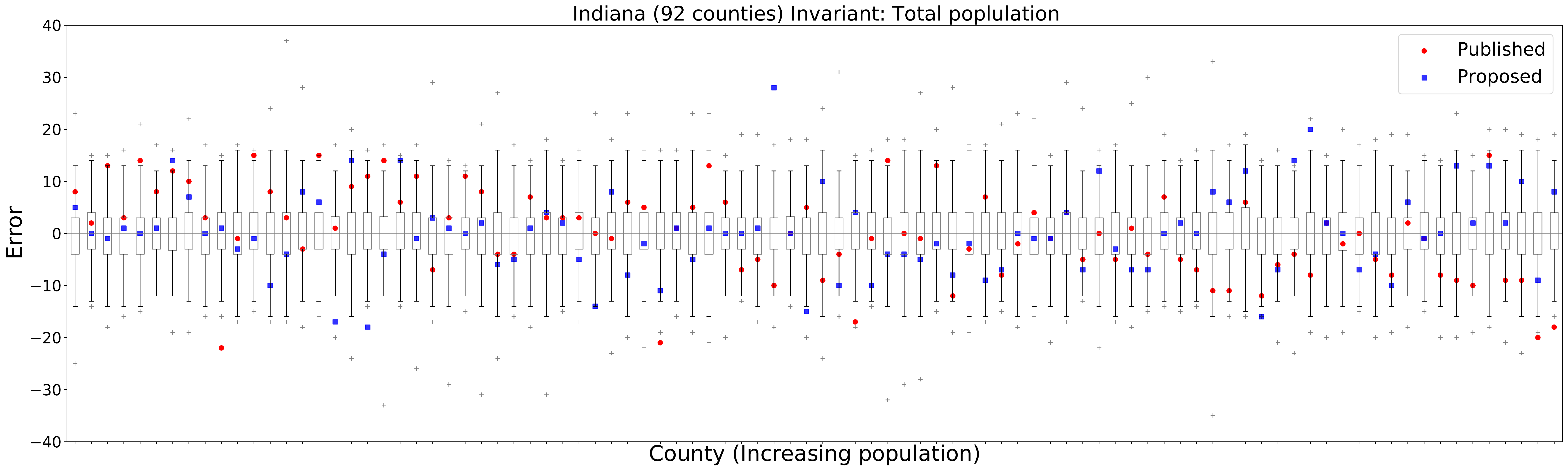} 
\includegraphics[width=.85\textwidth]{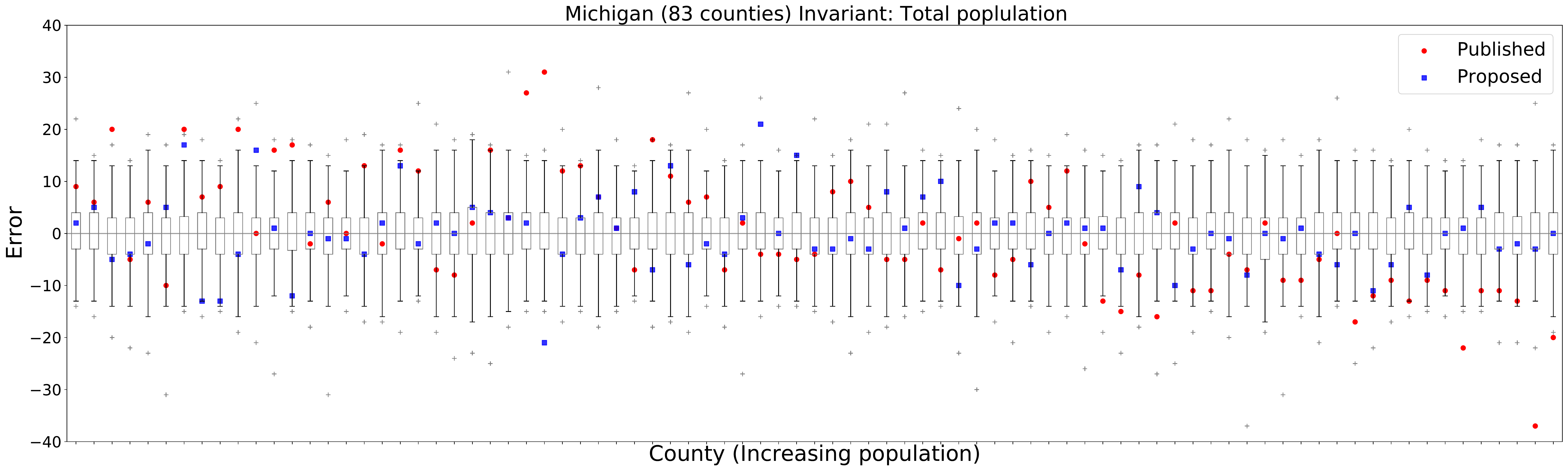}
\includegraphics[width=.85\textwidth]{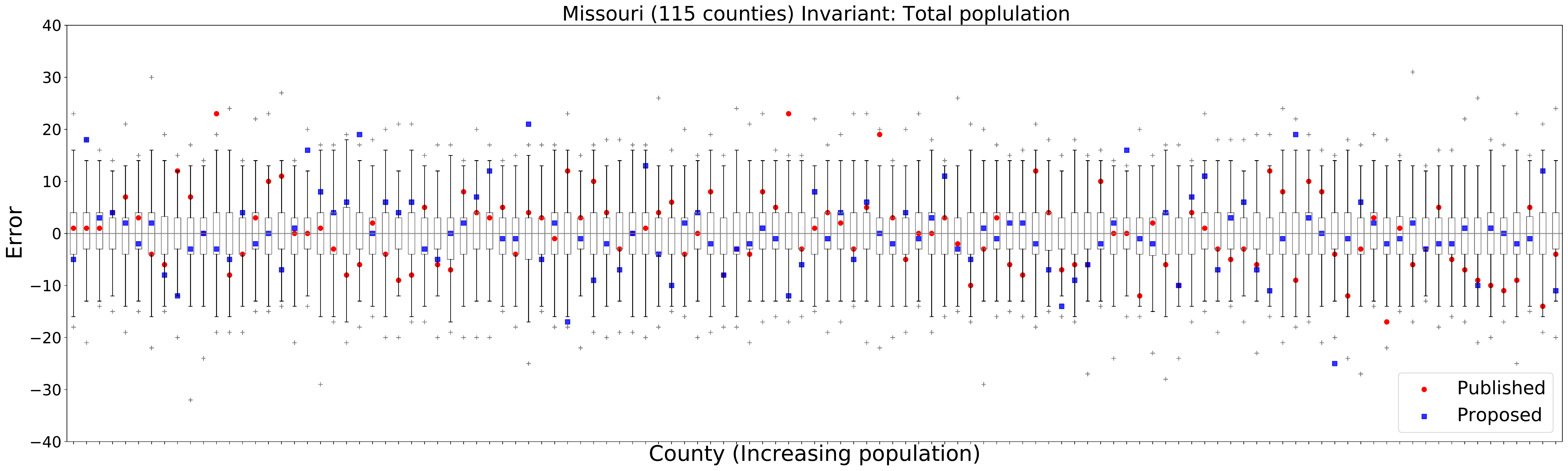}
\includegraphics[width=.85\textwidth]{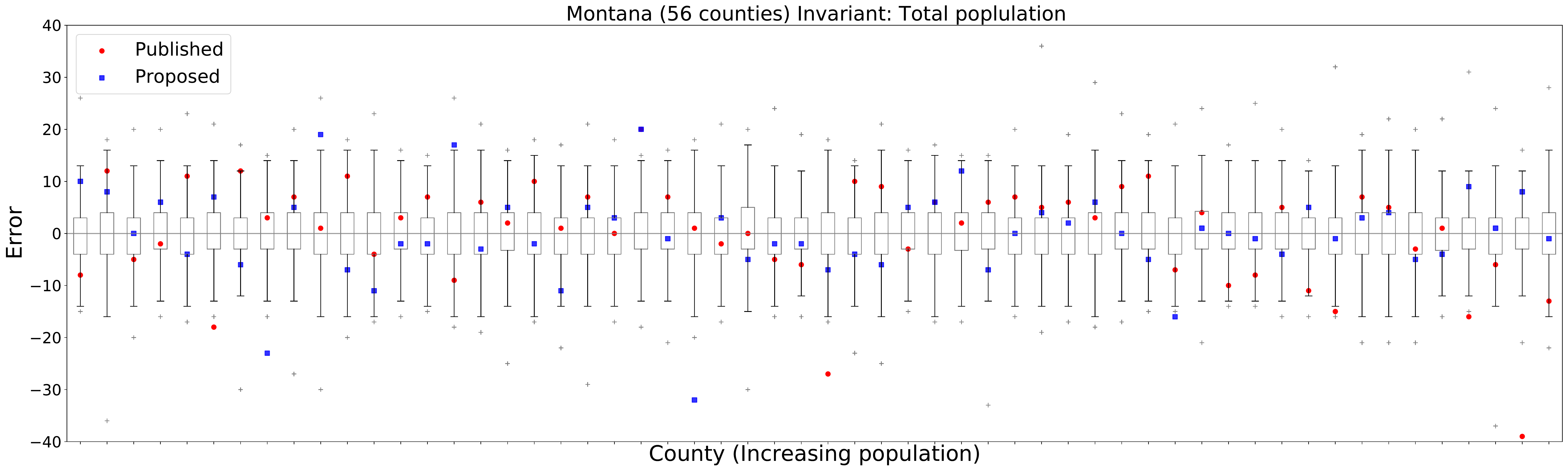}
\end{center}
\caption{\label{fig:census_sdp_county_combined} Privacy noise from DAS demonstration data (red dots) \cite[11/2020 vintage;][]{ipums2020das} vs. generalized Laplace mechanism (~\cref{def:lap}) via ~\cref{alg:metropolis-gibbs} (blue squares: one instance; boxplot: 1000 instances) for county populations within a state. The $x$-axes are arranged in increasing true county populations. State population total is invariant. The states shown here (plus Illinois in ~\cref{fig:illinois}) are those identified by~\citep{gao2022subspace} for which the DAS produced statistically significantly biased errors (at the $\alpha = 0.01$ level), as can be seen by the downwardly biased trends among the red. In contrast, the proposed noises are integer-valued and unbiased.}	
\end{figure*}
\begin{figure*}\ContinuedFloat
\begin{center}
    \includegraphics[width=.85\textwidth]{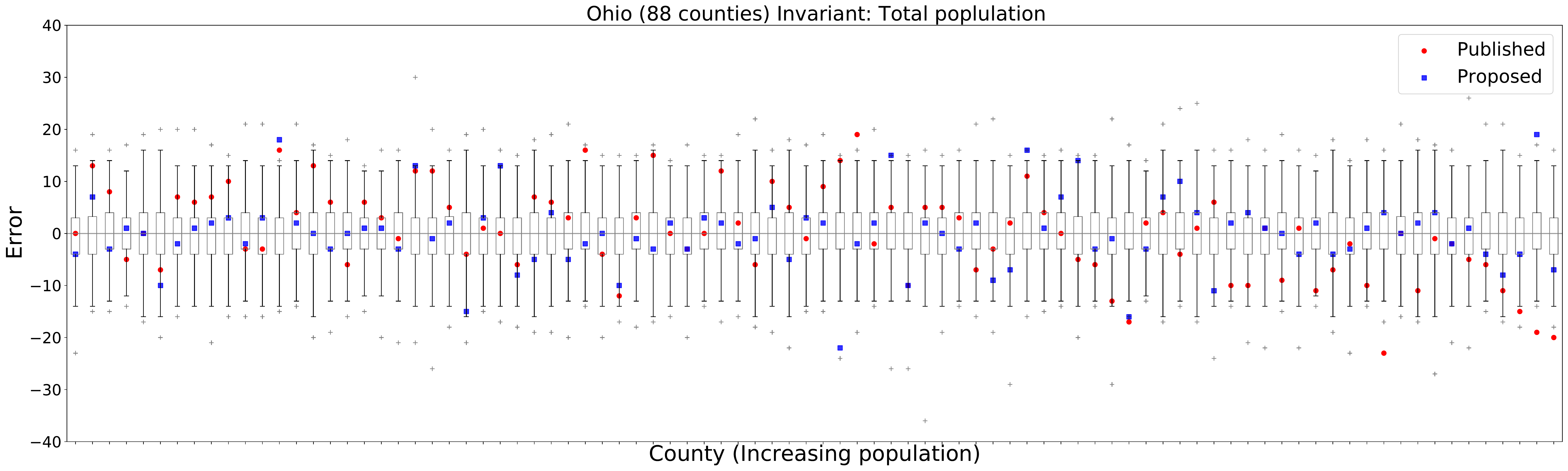} 
    \includegraphics[width=.85\textwidth]{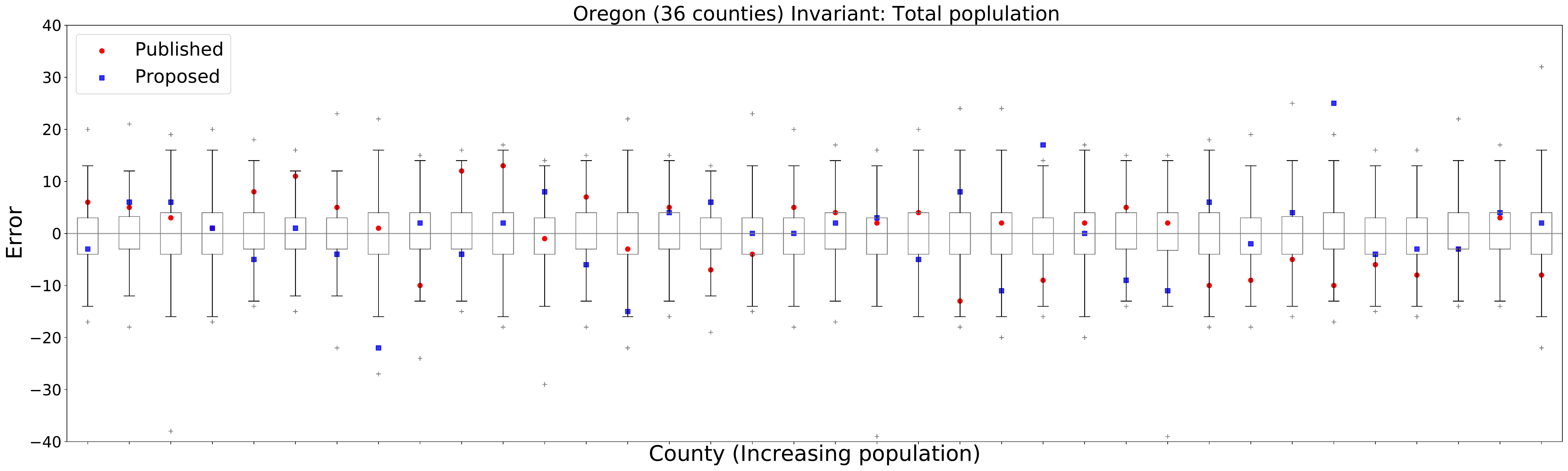} 
    \includegraphics[width=.85\textwidth]{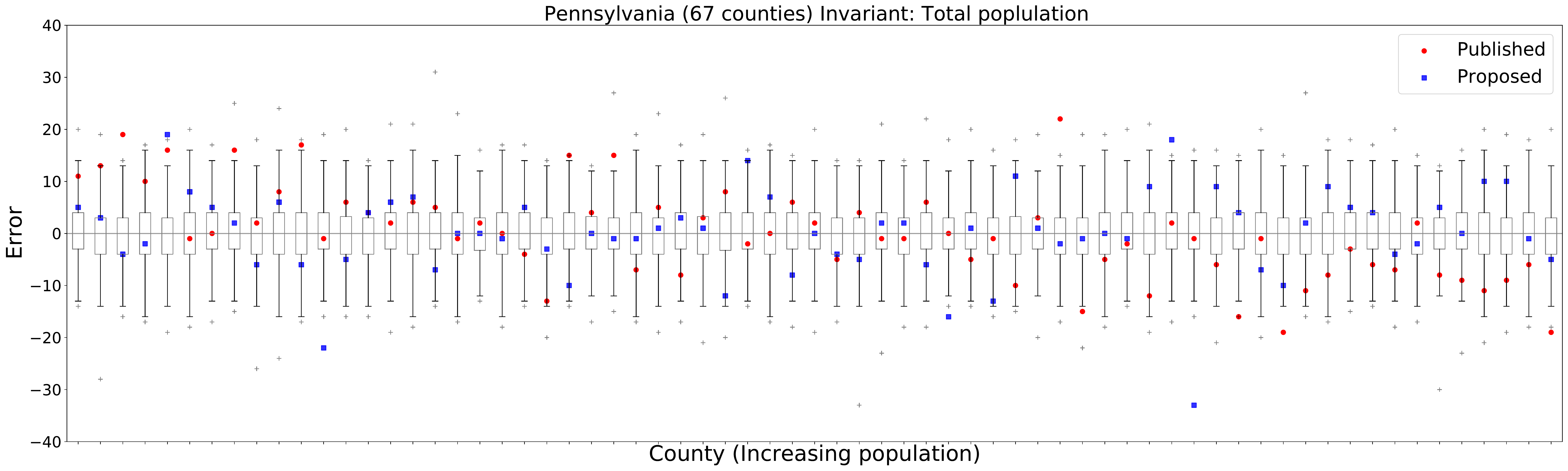} 
     \includegraphics[width=.85\textwidth]{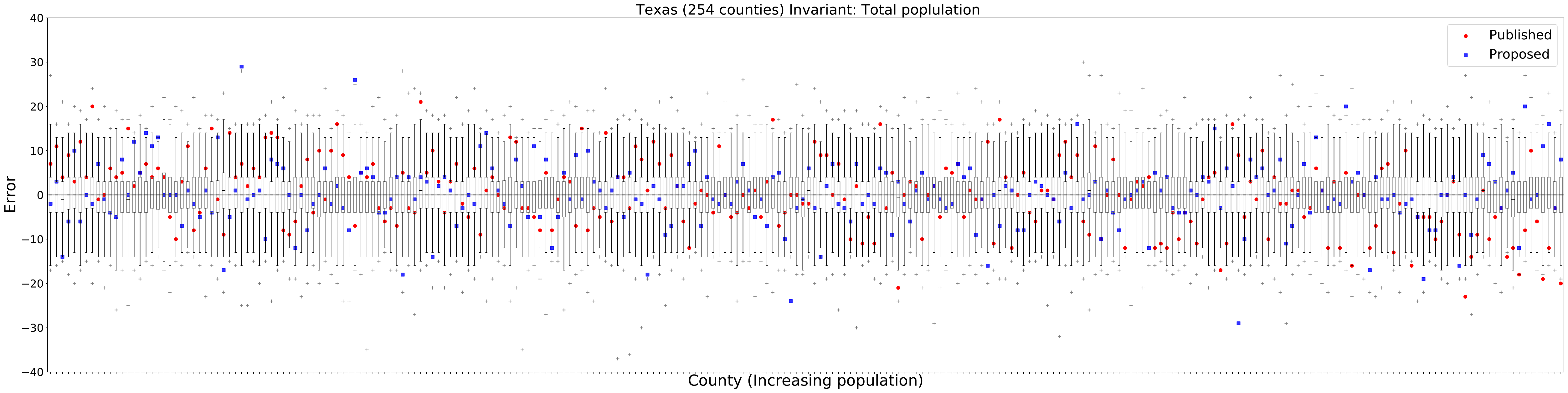} 
     \includegraphics[width=.85\textwidth]{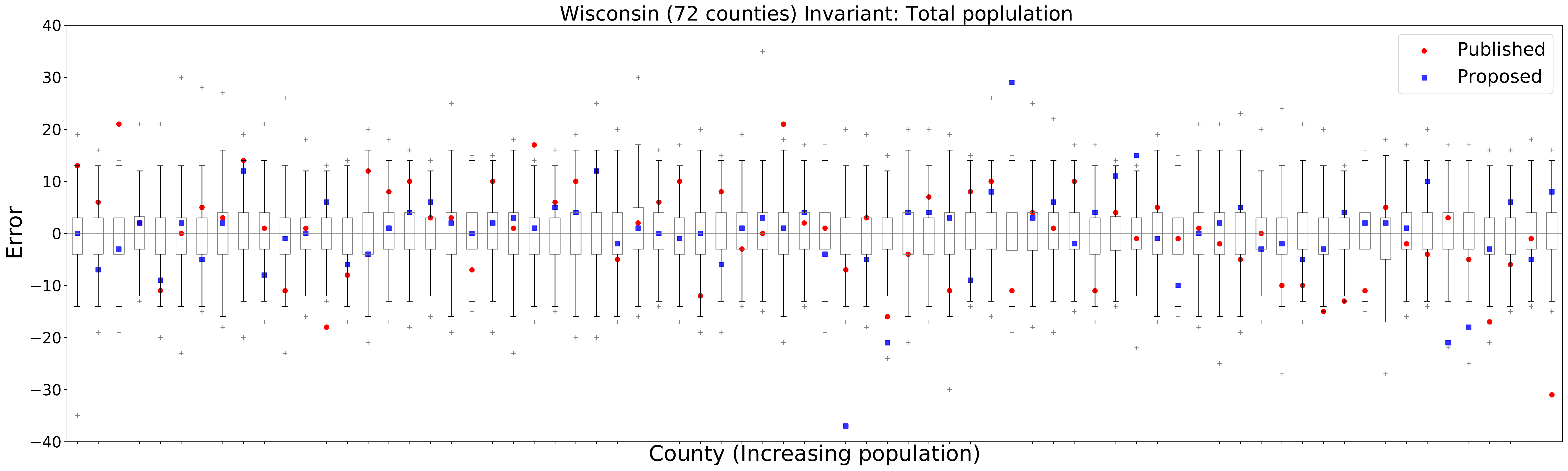}
    \caption{continued.}
\end{center}
\end{figure*}

\subsubsection{Comparison with projected discrete Laplace mechanism}

For publishing just the county level data under invariants, the most basic approach would be to add noise which ensures $\epsilon$-DP and then projecting onto space satisfying the constraints. As the latter portion of this procedure constitutes post-processing, the output is still private. One drawback with this approach is the procedure is now less transparent with regards to the randomness involved and the utility is now unclear due to optimization procedure required to project onto constraint satisfying subspace. Integer subspace differential privacy address both these concerns and though empirically noise magnitudes are similar we have a better understanding of the process compared to adding noise and then post-processing.

To compare, we consider databases that are at distance 1, i.e. $\|x-x'\|_1 = 1$ and compare generalized Laplace mechanism to $\epsilon-$DP mechanism for sum constraint (the sum of all entries is held invariant). The $\epsilon-$DP mechanism is obtained by adding Double geometric noise and then solving linear program to find closest vector in $\ell_2$ norm that satisfies the sum constraint. We then consider another vector closes in $\ell_1$ norm but with additional constraint that variables are now integers. We do not consider non-negativity constraints.  \cref{fig:census_sdp_county_combined_with_post} depicts in green - box plots per county for 1000 independent simulations of the above mentioned procedure whereas the setup for our mechanism is same as before. The box-plots majorly overlap indicating the magnitude of noise is comparable for both approaches.

\subsubsection{Comparison with TopDown Algorithm~\citep{abowd2022topdown}}

The TopDown algorithm, in addition to invariants and integer values, also performs post-processing to ensure that all released counts are non-negative. This requirement of non-negativity introduces  systematic bias into the privatized tabulations \citep{zhu2021bias,gao2022subspace}. 
\citet{hotz2020assessing} note that 
earlier versions of the DAS introduces inaccuracy in small-area data that compromised existing statistical analyses. An example may be seen from ~\cref{fig:illinois}: when the state population total is fixed, positive errors tend to associate with smaller county population counts, and negative errors with larger counts. 
Through six rounds of iterations
spanning 2019-2021, the bureau increased the privacy loss budget for the  PL. 94-171 files from $\epsilon=6$~\citep{census2019memo} to $(\epsilon, \delta)=(19.71, 10^{-10})$~\citep{census2021PLB}. 

When non-negativity is enforced, bias is theoretically unavoidable.
However, bias comes in a matter of degrees. The DAS demonstration dataset that we showcase here is tabulated at the county level, where the underlying microdata are dense and the confidential counts are far from zero. As a consequence, the non-negativity post-processing of counts privatized with moderately sized unbiased noise should not induce visible systematic bias as shown in ~\cref{fig:illinois}. To be precise, if we infuse the county-level counts with errors from the proposed generalized Laplace mechanism (shown in ~\cref{fig:illinois} and~\ref{fig:census_sdp_county_combined}) and then apply post-processing to truncate negative counts, none of the 1000 instantiations triggers the truncation in any county of any state. This is not surprising, since the smallest counties of most states have population sizes in the thousands (with the least populous county of the U.S. being Loving, Texas with a population of 82), whereas the overwhelming majority of our proposed noise instantiations are within [-30, 30] using the Bureau's privacy loss budget allocation for the county population table. Thus, even after nonnegativity post-processing is applied to our proposed solution, the resulting privatized counts are still effectively unbiased, because the post-processing applied to county-level tables is effectively vacuous.

A key design feature of the TopDown algorithm is that privatized counts at different geographic levels are produced separately, rather than only at the block level and aggregated all the way up. As Section 6.5 of~\citet{abowd2022topdown} explains, this takes advantage of the hierarchical structure of the Census geographies and facilitates fine-grained control over the Bureau's own fitness-for-use accuracy targets. While maintaining unbiasedness and nonnegativity is difficult at the block level, it is achievable at coarser levels as we demonstrate here. When the underlying geography is dense in counts, our proposal contributes a solution that delivers privatized tabulations that are effectively unbiased even if nonnegativity is furthermore imposed.

\begin{figure*}
\begin{center}
\includegraphics[width=.8\textwidth]{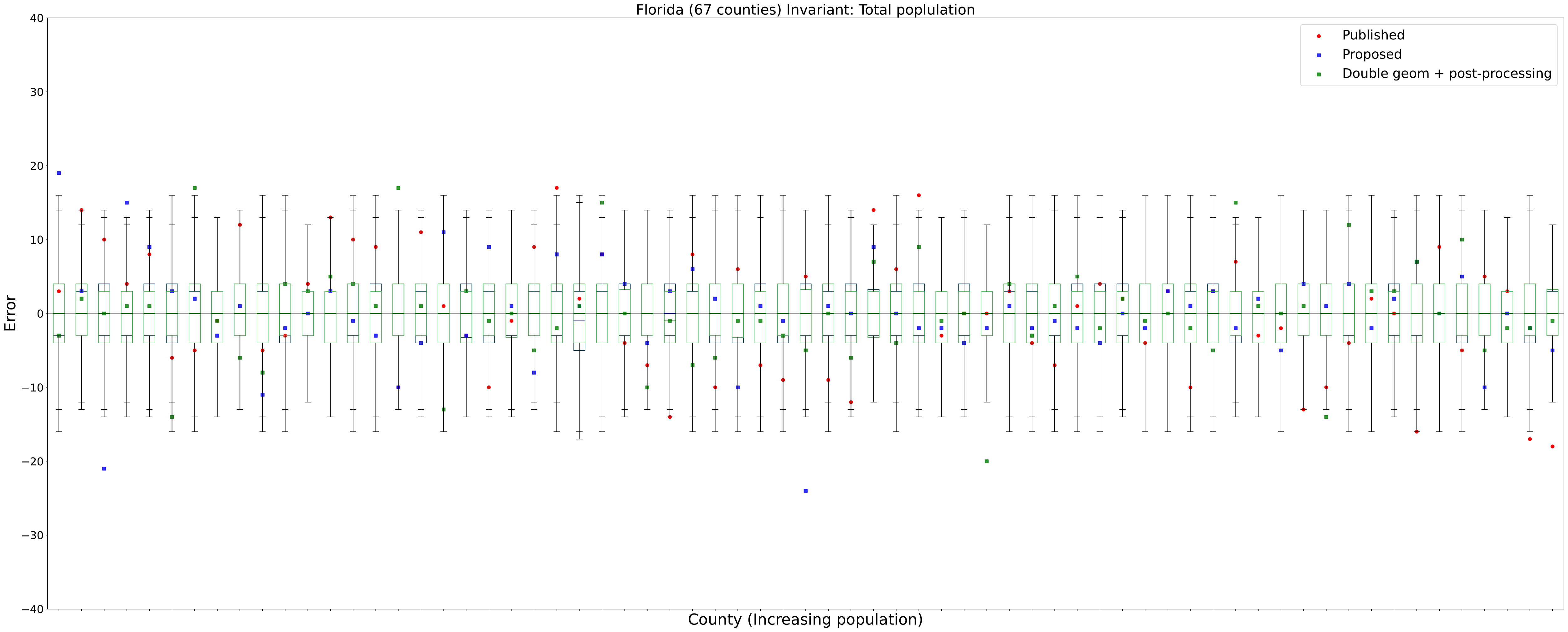}
\includegraphics[width=.8\textwidth]{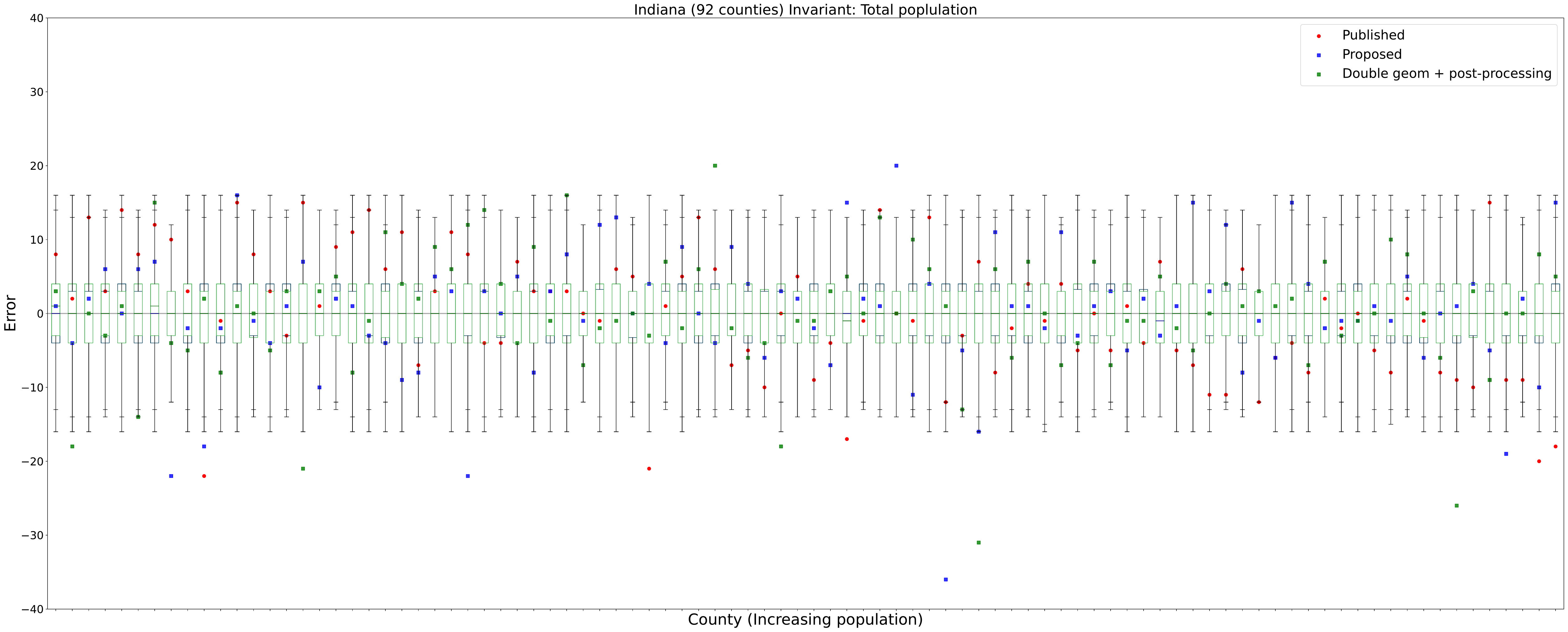} 
\includegraphics[width=.8\textwidth]{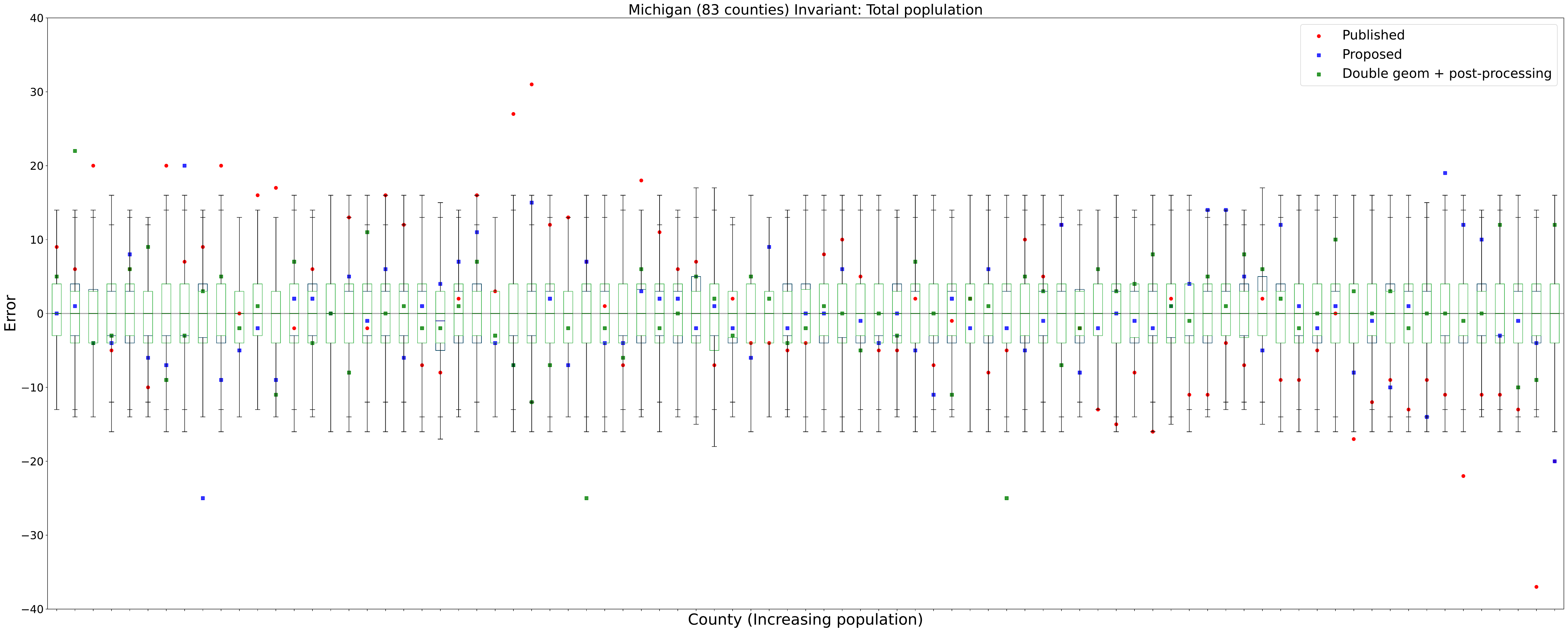}
\includegraphics[width=.8\textwidth]{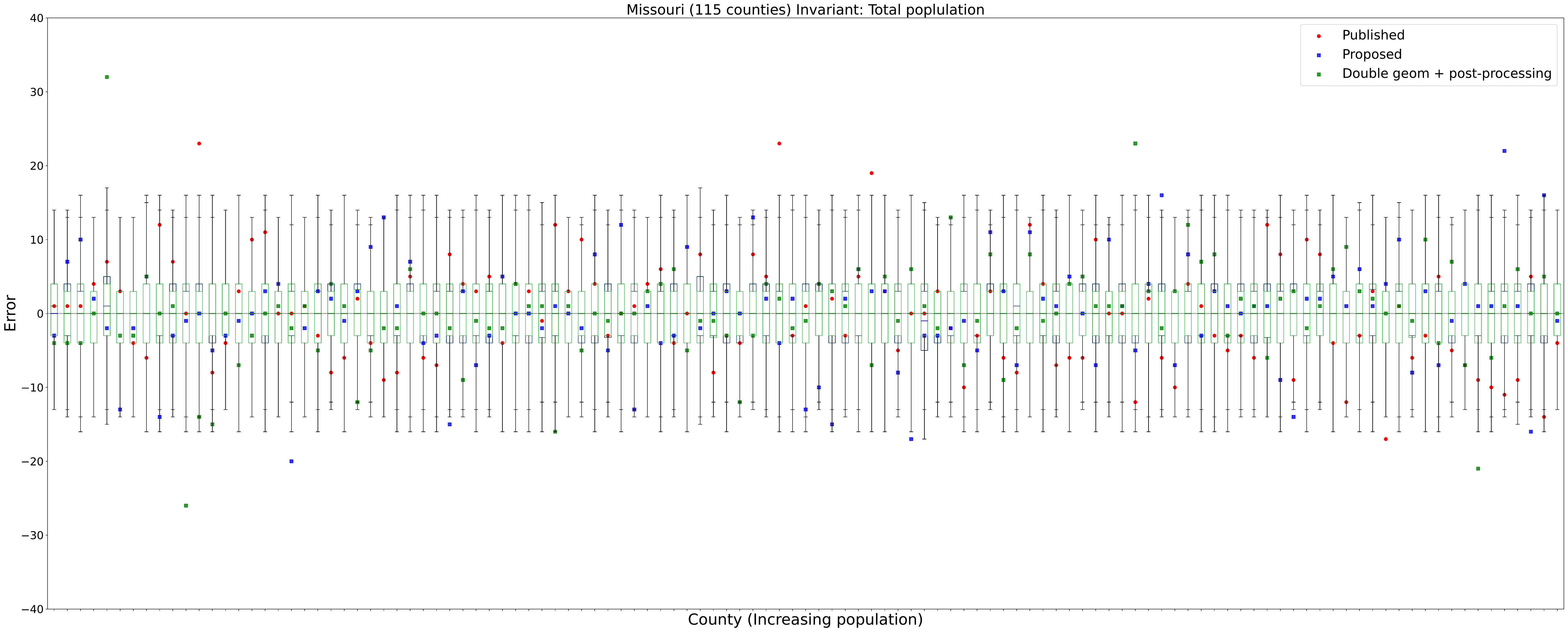}
\end{center}
\caption{\label{fig:census_sdp_county_combined_with_post} ~\cref{fig:census_sdp_county_combined} superimposed with additional post-processing on privacy noise generated from the generalized Laplace mechanism (green squares: one instance; boxplot: 1000 instances)  to enforce non-negativity of the privatized county-level population counts. In this case, post-processing did not alter the distribution of the original noise from the generalized Laplace mechanism (~\cref{def:lap} via ~\cref{alg:metropolis-gibbs}; blue squares: one instance; boxplot: 1000 instances) in any significant way, and did not induce a downward bias as seen from the privacy noise from DAS demonstration data (red dots) \cite[11/2020 vintage;][]{ipums2020das}.}	
\end{figure*}

\begin{figure*}\ContinuedFloat
\begin{center}
    \includegraphics[width=.8\textwidth]{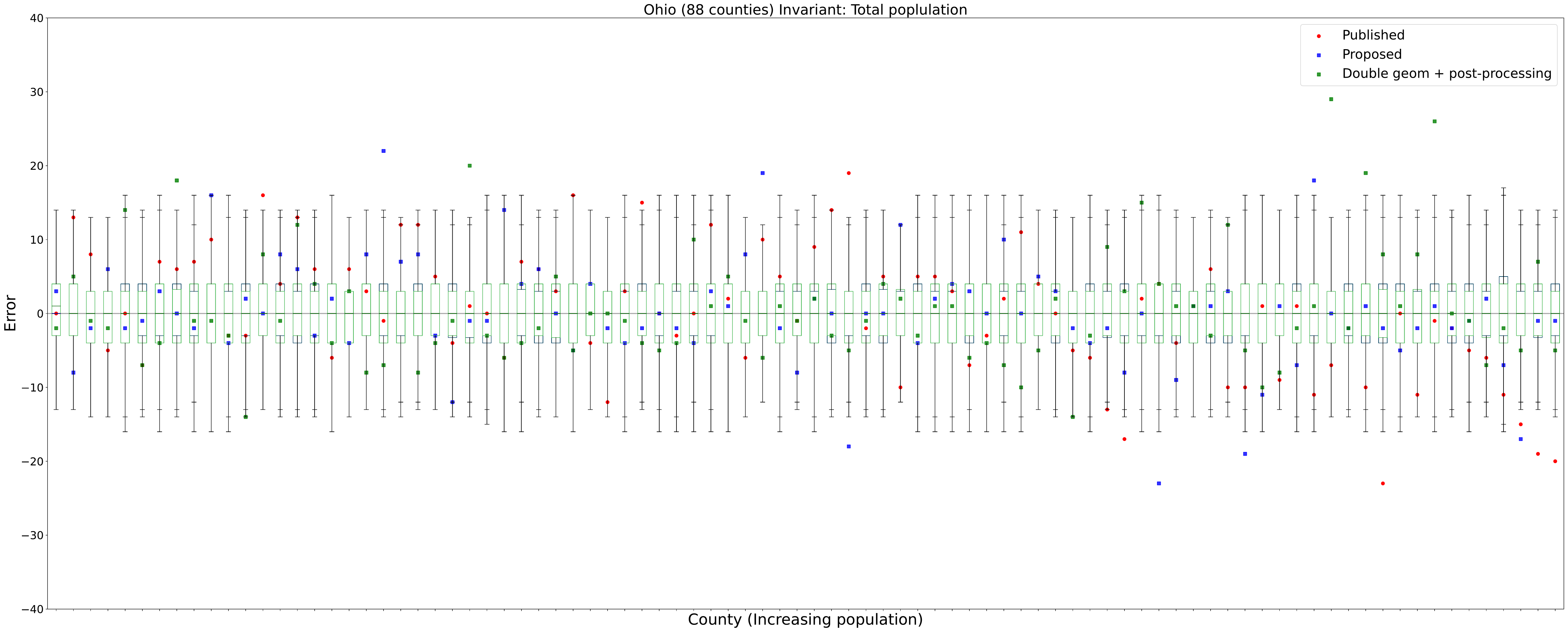} 
    \includegraphics[width=.8\textwidth]{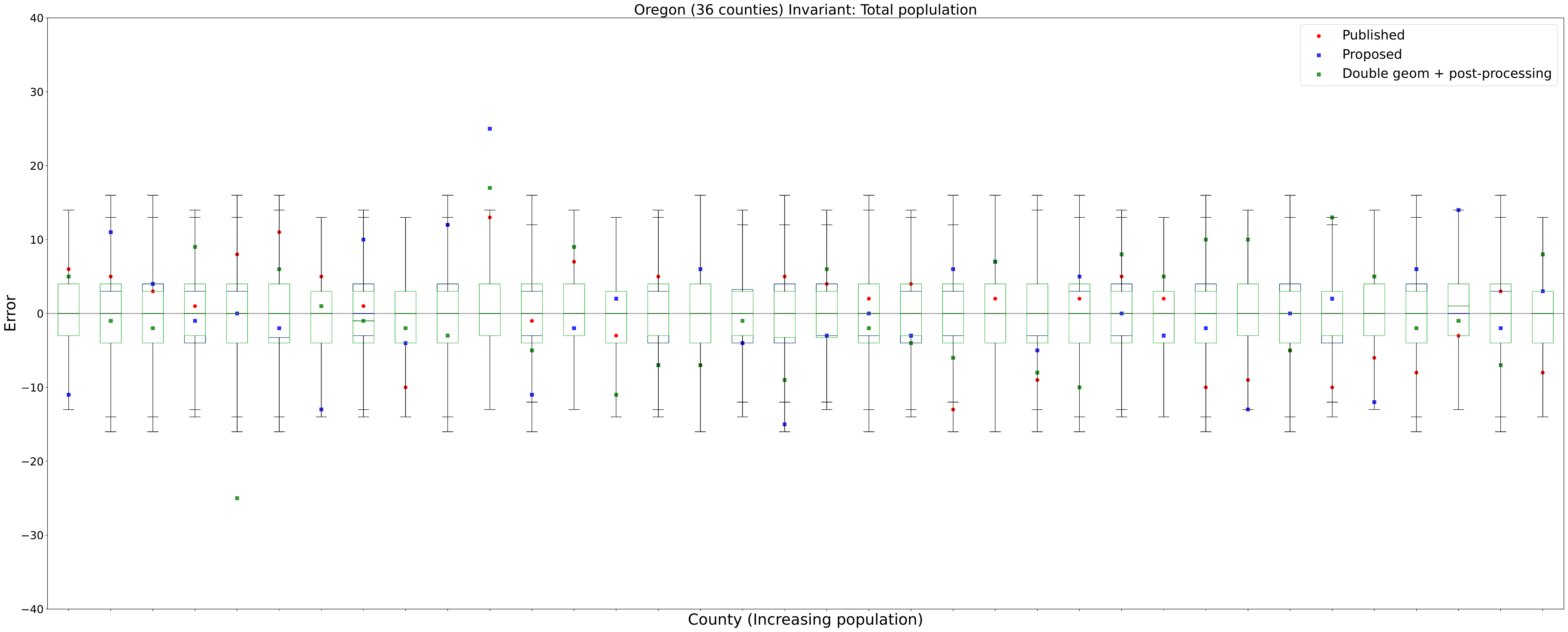} 
    \includegraphics[width=.8\textwidth]{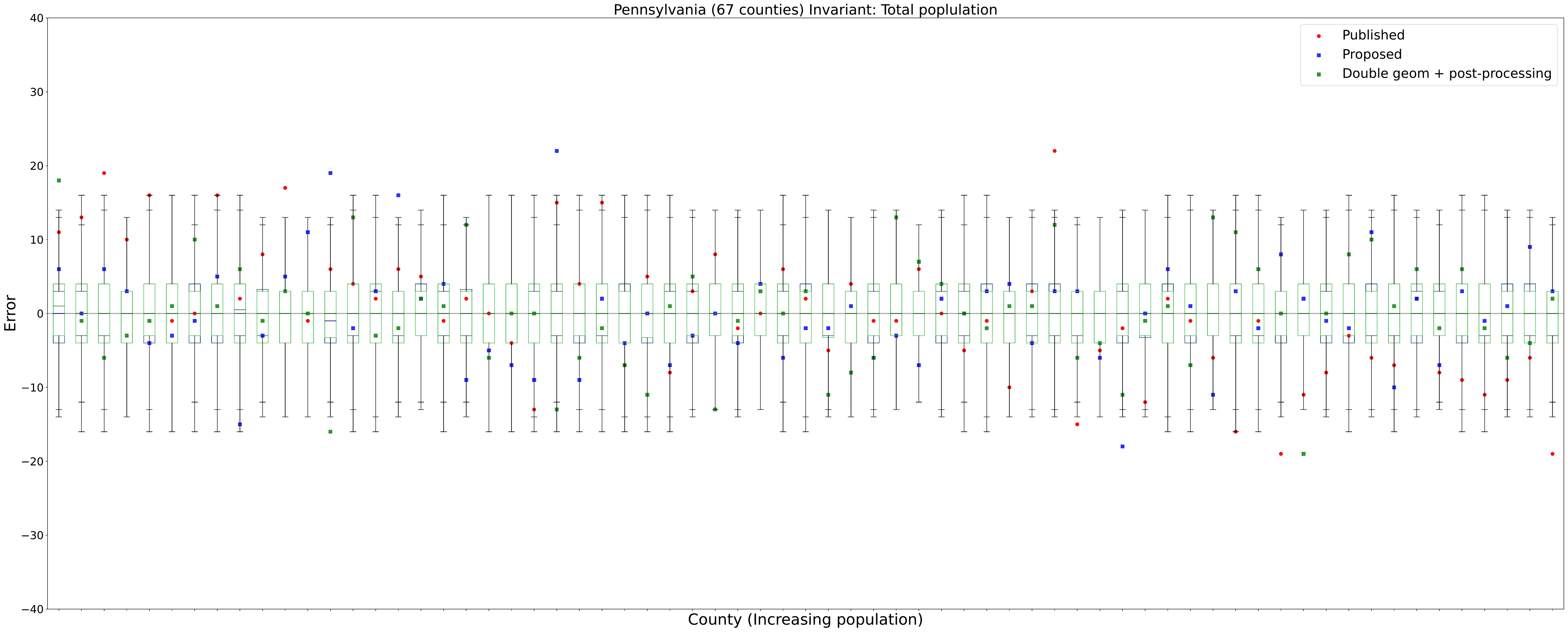} 
     \includegraphics[width=.8\textwidth]{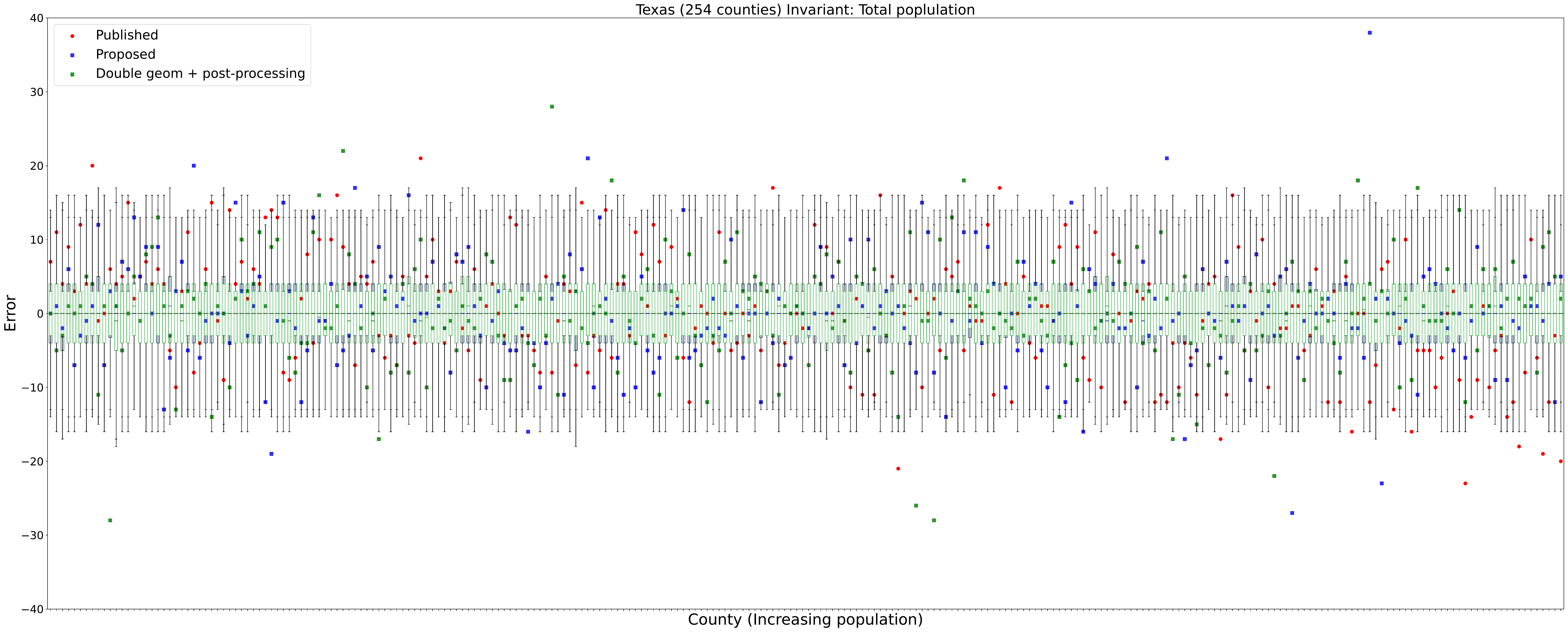} 
     \caption{continued.}
\end{center}
\end{figure*}

\begin{figure*}\ContinuedFloat
\begin{center}
    \includegraphics[width=.8\textwidth]{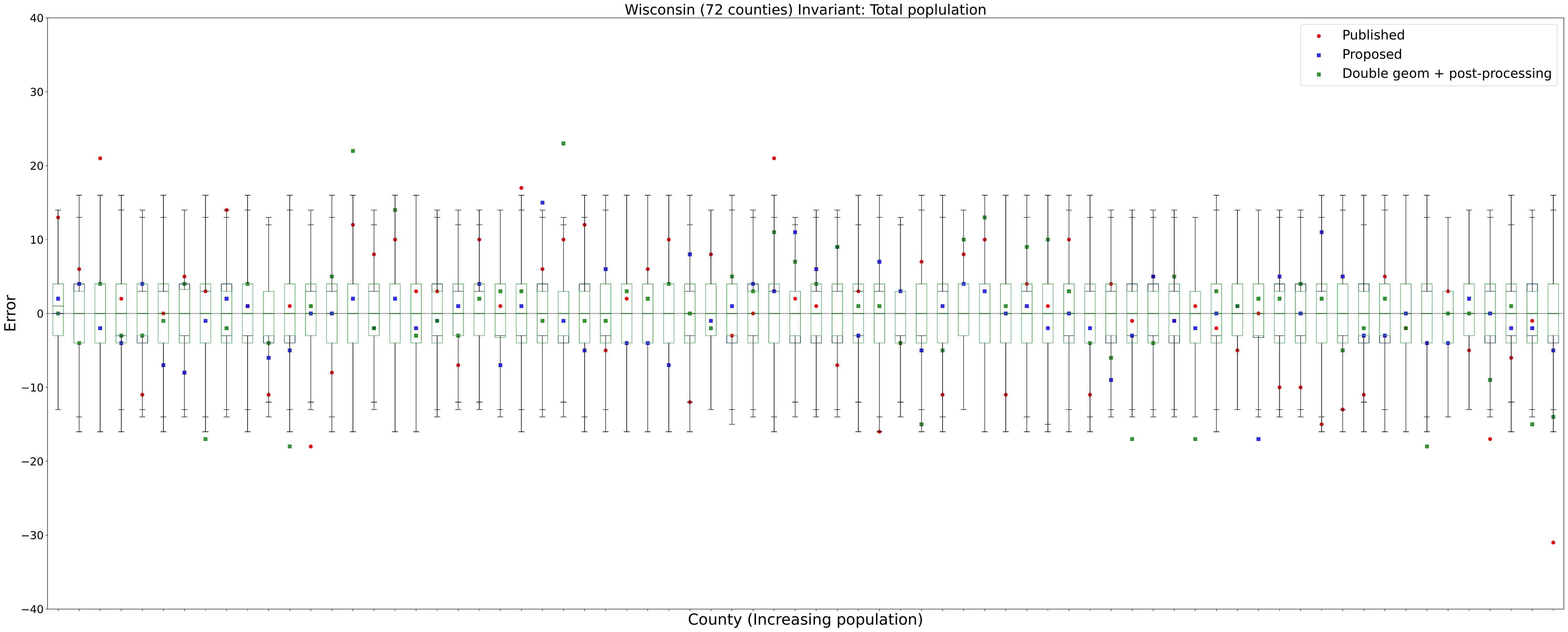}
    \includegraphics[width=.8\textwidth]{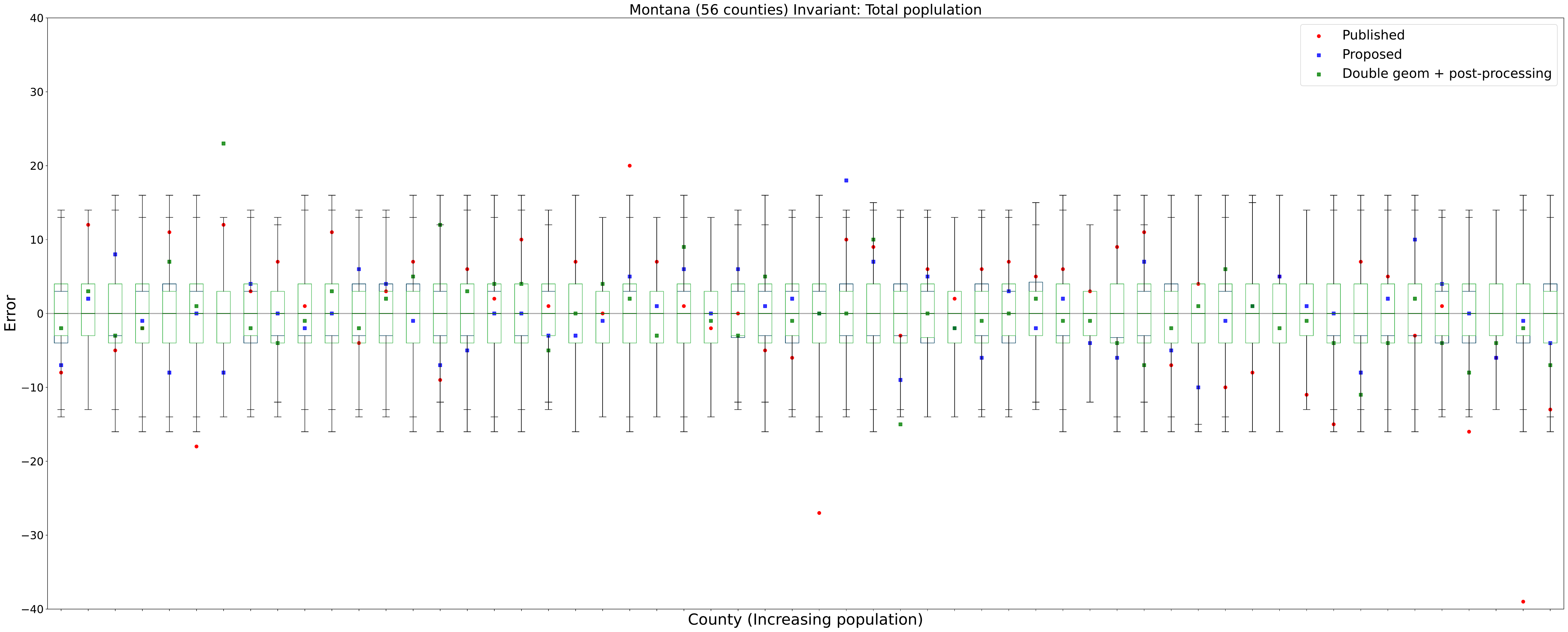}
    \includegraphics[width=.8\textwidth]{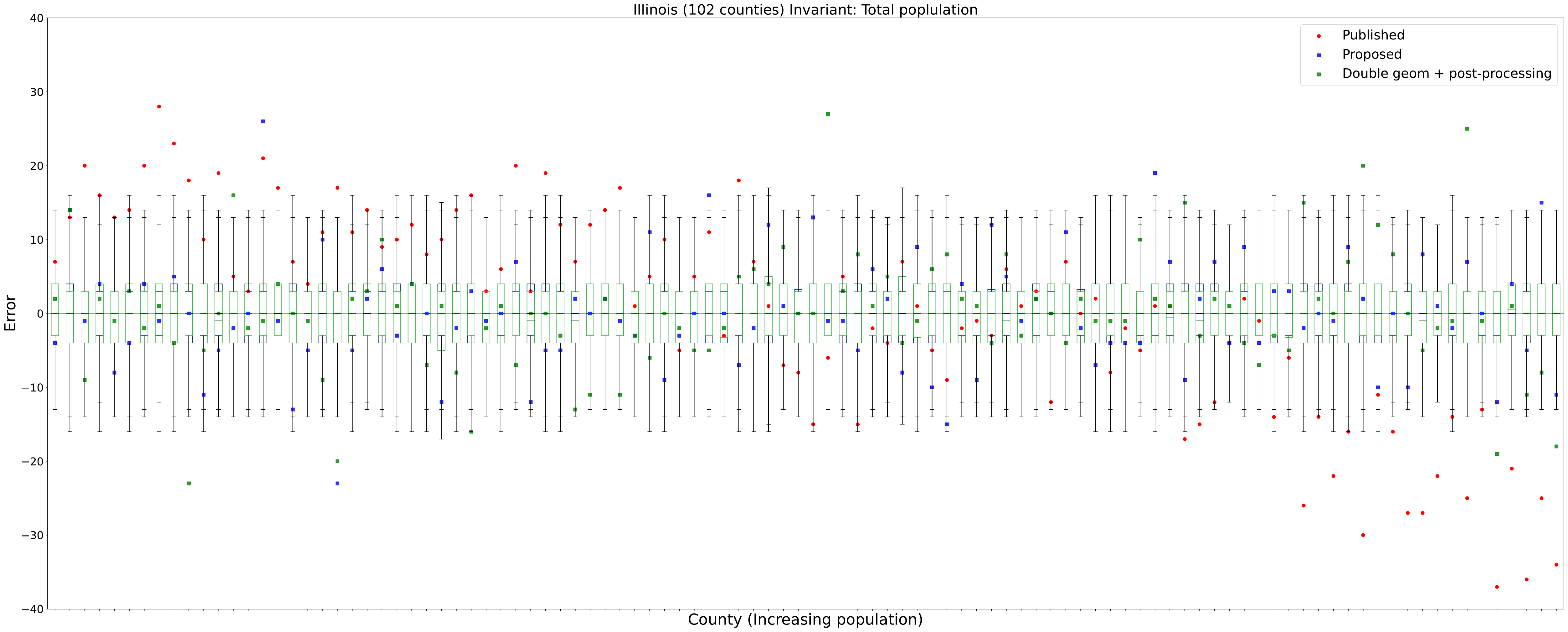}
    \caption{continued.}
\end{center}
\end{figure*}

\end{document}